\documentclass[pra,floatfix,nofootinbib,twocolumn]{revtex4}

\usepackage{amssymb}
\usepackage{graphicx}
\usepackage{amsmath}
\usepackage{amsfonts}

\newenvironment{proof}[1][Proof]{\noindent\textbf{#1.} }{\ \rule{0.5em}{0.5em}}

\newtheorem{myproposition}{Proposition}
\newtheorem{mylemma}{Lemma}
\newtheorem{mydefinition}{Definition}
\newtheorem{mytheorem}{Theorem}

\newtheorem{mycorollary}{Corollary}

\begin{document}

\title{Classical Ising model test for quantum circuits}
\author{Joseph Geraci}
\affiliation{Department of Mathematics, University of Toronto, Toronto, Ontario M5S 2E4,
Canada}
\affiliation{Center for Quantum Information Science and Technology, University of
Southern California, Los Angeles, CA 90089}
\altaffiliation{Current address: Ontario Cancer Institute, MaRS-TMDT, Toronto, Ontario M5G
1L7, Canada}
\author{Daniel A. Lidar}
\affiliation{Departments of Chemistry, Electrical Engineering, and Physics, Center for
Quantum Information Science and Technology, University of Southern
California, Los Angeles, CA 90089}

\begin{abstract}
We exploit a recently constructed mapping between quantum circuits and
graphs in order to prove that circuits corresponding to certain planar
graphs can be efficiently simulated classically. The proof uses an
expression for the Ising model partition function in terms of quadratically
signed weight enumerators (QWGTs), which are polynomials that arise
naturally in an expansion of quantum circuits in terms of rotations
involving Pauli matrices. We combine this expression with a known efficient
classical algorithm for the Ising partition function of any planar graph in
the absence of an external magnetic field, and the Robertson-Seymour theorem
from graph theory. We give as an example a set of quantum circuits with a
small number of non-nearest neighbor gates which admit an efficient
classical simulation.
\end{abstract}

\maketitle


\section{Introduction}

From its early days quantum computing was perceived as a means to
efficiently simulate physics problem \cite{Feynman1,Lloyd:96}, and a host of
results have been derived along these lines for quantum \cite%
{Wiesner:96,Meyer:97,Boghosian:97a,Abrams:97,Zalka:98,Lidar:98RC,Ortiz:00,Terhal:00,Freedman:00,WuByrdLidar:01,Aspuru-Guzik:05,Cirac:02,Cirac:08}%
, and classical systems \cite%
{Lidar:PRE97a,Yepez:01,Meyer:02,Georgeot:01a,Georgeot:01b,Terraneo:03,Lidar:04,dorit-tutte,JOE,Bravyi:07,aspuru-guzik1,aspuru-guzik2}%
. A\ natural problem relating quantum computation and statistical mechanics
is to understand for which instances quantum computers provide a speedup
over their classical counterparts for the evaluation of partition functions 
\cite{Lidar:PRE97a,Lidar:04}. For the \emph{Potts model}, results obtained
in \cite{dorit-tutte} provide insight into this problem when the evaluation
is an additive approximation. We provided a class of examples for which
there is a quantum speedup when one seeks an exact evaluation of the Potts
partition function \cite{JOE}.

In this work we address the connection between quantum computing and
classical statistical mechanics from the opposite perspective. Namely, we
seek to find restrictions on the power of quantum computing, by employing
known results about efficiently simulatable problems in statistical
mechanics. Specifically, we restrict our attention to the \emph{Ising model}
partition function $Z$, and use a mapping between graph instances of the
Ising model and quantum circuits introduced in \cite{JOE2}, to identify a
certain class of quantum circuits which have an efficient classical
simulation.

Restricted classes of quantum circuits which can be efficiently simulated
classically have been known since the Gottesman-Knill theorem \cite%
{Nielsen:book}. This theorem states that a quantum circuit using only the
following elements can be simulated efficiently on a classical computer:\
(1) preparation of qubits in computational basis states, (2) quantum gates
from the Clifford group (Hadamard, controlled-NOT gates, and Pauli gates),
and (3) measurements in the computational basis. Such \textquotedblleft
stabilizer circuits\textquotedblright\ on $n$ qubits can be be simulated in $%
O(n\log n)$ time using the graph state formalism \cite{anders:022334}. Other
early results include Ref. \cite{valiant}, where the notion of matchgates
was introduced and the problem of efficiently simulating a certain class of
quantum circuits was reduced to the problem of evaluating the Pfaffian. This
was subsequently shown to correspond to a physical model of noninteracting
fermions in one dimension, and extended to noninteracting fermions with
arbitrary pairwise interactions \cite{Knill:01,Terhal:01,Terhal:05} (see
further generalizations in Refs. \cite{Bravyi:05,jozsa-2008}), and
Lie-algebraic generalized mean-field Hamiltonians \cite{somma:190501}.
Criteria for efficient classical simulation of quantum computation can also
be given in terms of upper bounds on the amount of entanglement generated in
the course of the quantum evolution \cite{vidal}.

A result that is more directly related to the one we shall present in this
work is given in Ref. \cite{Bravyi:07}, but within the measurement-based
quantum computation (MQC) paradigm. MQC relies on the preparation of a
multi-qubit entangled resource state known as the cluster state. It is known
that MQC with access to cluster states is universal for quantum computation.
Reference \cite{Bravyi:07} considers \emph{planar code states} which are
closely related to cluster states in that a sequence of Pauli-measurements
applied to the two-dimensional cluster state can result in a planar code
state. MQC with planar code states consists of a sequence of measurements $%
\{M_{1},M_{2},\dots ,M_{n},M\}$ where the $M_{i}$ are one-qubit measurements
and $M$ is a final measurement done on the remaining qubits in some basis
which depends on the results of the $M_{i}$. Reference \cite{Bravyi:07}
demonstrates that planar code states are not a sufficient resource for
universal quantum computation (and can be classically simulated). This fact
is attributed to the exact solvability of the Ising partition function on
planar graphs. Our results complement the work in \cite{Bravyi:07}, as they
are provided in terms of the circuit model, and generalize to Ising model
instances that correspond to graphs which are not necessarily subgraphs of a
two-dimensional grid.

Other conceptually related work uses the connection between graphs and
quantum circuits and the formalism of tensor network contractions, to show
that any polynomial-sized quantum circuit of $1$- and $2$- qubit gates,
which has log depth and in which the $2$-qubit gates are restricted to act
at bounded range, may be classically efficiently simulated \cite%
{Markov:tensor,jozsa-2008,yoran:170503}. A tensor network is a product of
tensors associated with vertices of some graph $G$ such that every edge of $%
G $ represents a summation (contraction) over a matching pair of indexes. We
also use a relationship between quantum circuits and graphs but whose
construction is quite different \cite{JOE2}. Also, Ref. \cite{Bravyi:08}
connects matchgates and tensor network contractions to notions of efficient
simulation.

Finally, other closely related work was recently reported in \cite{Nest:08}
(see also \cite{nest:117207,nest:110501,Cuevas:08}), which addresses the
classical simulatability of quantum circuits. Their results use a connection
to the partition function of spin models, as do we, and they too provide a
mapping between classical spin models and quantum circuits. Specifically
pertinent to our work is the fact that they give criteria for the
simulatability of quantum circuits, using the 2D Ising model. That is,
circuits consisting of single qubit gates of the form $e^{i\theta \sigma
_{x}}$ and nearest-neighbor gates of the form $e^{i\phi \sigma _{z}\otimes
\sigma _{z}}$ are classically efficiently simulable. We shall discuss
how the nearest-neighbor restrictions can be lifted while retaining
efficient classical simulatability.

The structure of this paper is as follows. We begin with a brief review of
the Ising model in Section \ref{sec:Ising}, where we define the Ising
partition function $Z$. In Section \ref{sec2} we review quadratically signed
weight enumerators (QWGT's) and their relationship to quantum circuits, and
review the relationship between QWGT's and $Z$. In Section \ref{sec:mapping}
we introduce an ansatz that allows one to associate graph instances of the
Ising model with circuit instances of the quantum circuit model. In this
section we derive a key result: an explicit connection between the partition
function for the Ising model on a graph, and a matrix element of the unitary
representing a quantum circuit which is related to this graph via the
graph's incidence matrix [Eq.\emph{\ }(\ref{eq:element})]. We then present
our main result in Section \ref{sec:proof}:\ a theorem on efficiently
simulatable quantum circuits. The proof depends on the fact that there are
algorithms for the efficient evaluation of $Z$ for planar instances of the
Ising model. We also discuss the relation to previous work. In Section \ref%
{nextstep} we present a discussion and some suggestions for future work,
including the possibility of a quantum algorithm for the additive
approximation of $Z$. We conclude in Section \ref{sec:conc}. The Appendix
gives a review of pertinent concepts from graph theory, and additional
details, including some proofs.

\section{Ising Spin Model}

\label{sec:Ising}

We briefly introduce the Ising spin model accompanied by some notation and
definitions. Let $G=(E,V)$ be a finite, arbitrary undirected graph with $|E|$
edges and $|V|$ vertices. In the Ising model each vertex $i$ is occupied by
a classical spin $\sigma _{i}=\pm 1$, and each edge $(i,j)\in E$ represents
a bond $J_{ij}$ (interaction energy between spins $i$ and $j$).

\begin{mydefinition}
\label{def:Ising}An \emph{instance} of the Ising problem is the data $\Delta
\equiv (G,\{J_{ij}\})$, i.e., $\Delta $ represents a weighted graph.
\end{mydefinition}

The Hamiltonian of the spin system is 
\begin{equation}
H_{\Delta }(\sigma )=-\sum_{(i,j)\in E}J_{ij}\sigma _{i}\sigma _{j}.
\label{Ham}
\end{equation}%
A spin configuration $\sigma =\{\sigma _{i}\}_{i=1}^{|V|}$ is a particular
assignment of spin values for all $|V|$ spins. A bond with $J_{ij}>0$ is
called ferromagnetic, and a bond with $J_{ij}<0$ is called
antiferromagnetic. The probability of the spin configuration $\sigma $ in
thermal equilibrium for a system in contact with a heat reservoir at
temperature $T$, is given by the Gibbs distribution: $P_{\Delta }(\sigma )={%
\frac{1}{Z_{\Delta }}}W_{\Delta }(\sigma )$, where the Boltzmann weight is $%
W_{\Delta }(\sigma )=\exp [-\beta H_{\Delta }(\sigma )]$, $\beta =1/kT$
is the inverse temperature in energy units, $k$ is the Boltzmann
constant, and $Z_{\Delta }$ is the partition function: 
\begin{equation}
Z_{\Delta }(\beta )\equiv \sum_{\sigma }\exp [-\beta H_{\Delta }(\sigma )].
\label{eq:Z}
\end{equation}%
(Unless there is a risk of confusion we will from now on write $Z$ in
place of $Z_{\Delta }(\beta )$ in order to simplify our notation.)
Computation of the partition function is the canonical problem of
statistical mechanics, since once $Z$ is known one can compute all
thermodynamic quantities, such as the magnetization and heat capacity, by
taking derivatives of $F=-k\log Z$ (the free energy)\ with respect to
appropriate thermodynamic variables \cite{Reichl:book}.

In this work we restrict our attention to the case $J_{ij}\in
\{-J,0,J\}$, with $J>0$, which already gives rise to the full
complexity of spin glass 
models and the associated computational hardness \cite{Parisi:book}. For
example, with the above restriction the problem of computing the partition
function in the three-dimensional spin-glass is NP-hard \cite{Barahona}.%
\footnote{%
A problem is called NP-hard if the existence of a polynomial-time algorithm
for its solution implies the existence of such an algorithm for all
NP-complete problems.}

\section{Quadratically Signed Weight Enumerators and their Relation to the
Ising Partition Function}

\label{sec2}

Quadratically Signed Weight Enumerators (QWGTs) were introduced by Knill and
Laflamme in Ref. \cite{Laflamme}.

\begin{mydefinition}
A Quadratically Signed Weight Enumerator is a bi-variate polynomial of the
form 
\begin{equation}
S(A,B,x,y)=\sum_{b\in \ker A}(-1)^{b^{t}Bb}x^{|b|}y^{n-|b|},  \label{eq:S}
\end{equation}%
where $A$ and $B$ are $0,1$-matrices with $B$ of dimension $n\times n$ and $A
$ of dimension $m\times n$. The variable $b$ in the summand ranges over $0,1$%
-column vectors of dimension $n$ satisfying $Ab=0$ (in the kernel, or
nullspace of $A$), $b^{t}$ is the transpose of $b$, and $|b|$ is the Hamming
weight of $b$ (the number of ones in the vector $b$). All calculations
involving $A,B$ or $b$ are done modulo $2$.
\end{mydefinition}

Note that the evaluation of a QWGT, given that $x$ and $y$ are natural
numbers, is in general $\#$P-hard, since it includes the evaluation of the
weight enumerator polynomial of a classical linear code \cite{Welsh:book}.

\subsection{QWGTs from Quantum Circuits}

\label{sec:QWGT-QC}

We shall now review in some detail how QWGT's were arrived at in Ref. \cite%
{Laflamme} by considering expansions of quantum circuits. Let $\Omega $ be a
quantum circuit formed by a temporal ordering of $N$ gates $g_{k}$, and let $%
U(\Omega )=\prod\nolimits_{k=N}^{1}g_{k}=g_{N}\cdots g_{1}$ be the
corresponding unitary operator. Note that a universal gate set can be
achieved by allowing arbitrary rotations about tensor products of Pauli
operators, i.e., each of the $N$ gates $g_{k}$ can be represented as

\begin{equation}
e^{-i\sigma _{b}\theta /2}=\cos \left( \frac{\theta }{2}\right) I-i\sin
\left( \frac{\theta }{2}\right) \sigma _{b},  \label{eq:G}
\end{equation}%
where 
\begin{equation}
\sigma _{b}=\bigotimes_{i=1}^{n}\sigma _{b_{i}}^{(i)},
\end{equation}%
with $n$ being the number of qubits, such that the Pauli matrices are 
\begin{eqnarray*}
\sigma _{00} &=&I=\left( 
\begin{array}{cc}
1 & 0 \\ 
0 & 1%
\end{array}%
\right) ,\quad \sigma _{01}=\sigma _{X}=\left( 
\begin{array}{cc}
0 & 1 \\ 
1 & 0%
\end{array}%
\right) , \\
\sigma _{11} &=&\sigma _{Y}=\left( 
\begin{array}{cc}
0 & -i \\ 
i & 0%
\end{array}%
\right) ,\quad \sigma _{10}=\sigma _{Z}=\left( 
\begin{array}{cc}
1 & 0 \\ 
0 & -1%
\end{array}%
\right) .
\end{eqnarray*}%
Here $b_{i}\in \{00,01,10,11\}$, $b=\{b_{i}\}_{i=1}^{n}$ is a binary vector
whose length is $2n$, and the superscript $(i)$ represents the qubit which
is operated on by the corresponding Pauli matrix. A circuit constructed
using gates of the form (\ref{eq:G}) may be approximated efficiently to
accuracy $O(\epsilon /N)$ with $\mathrm{polylog}(N/\epsilon )$ overhead
using a standard gate set, such as controlled-NOT\ with single-qubit gates,
and there is a classical algorithm that computes such approximations
efficiently \cite{Laflamme,Kitaev:96}. A universal set of one- and two-qubit
gates can be obtained from $g_{k}$'s as in Eq.~(\ref{eq:G}) from the
rotations with $\cos (\theta /2)=3/5$ (i.e., $\theta =2\arcsin (4/5)$)
around operators of weight at most two (the weight of $\sigma _{b}$ is the
number of non-zero pairs of bits in $b$) (Theorem 3.3, case (e), of \cite%
{Adleman:97}). Letting 
\begin{equation}
\cos \left( \frac{\theta }{2}\right) =\frac{\alpha }{\gamma },\quad \sin
\left( \frac{\theta }{2}\right) =\frac{\alpha ^{\prime }}{\gamma },
\label{eq:theta}
\end{equation}%
so that $\gamma =\sqrt{\alpha ^{2}+\alpha ^{\prime 2}}$, we rewrite Eq.~(\ref%
{eq:G}) as 
\begin{equation}
g_{k}=\frac{1}{\gamma }\left( \alpha I-i\alpha ^{\prime }\sigma
_{b_{k}}\right) .  \label{eq:Gk}
\end{equation}

The gate set is still universal if $U(\Omega )$ is expressed as a product of 
\emph{real} gates \cite{Bernstein:93}, i.e., if each gate $g_{k}$ contains
an odd number of $\sigma _{Y}$'s, so that $i\sigma _{b_{k}}$ in Eq. (\ref%
{eq:Gk}) is a real-valued matrix. Following Ref. \cite{Laflamme}, we
adopt this convention, so that from now on $b_{k}$ is a binary vector of
length $2n$, subject to the restriction that the $b_{k}$ can only contain an
odd number of $11$'s. Moreover, the gate set is still universal if we assume
that the orientation (the sign of $\theta $) is positive if the number of $%
\sigma _{Y}$'s is $1\mathrm{mod}(4)$ and negative otherwise \cite{Laflamme}.
This means that we can replace Eq.~(\ref{eq:Gk}) by%
\begin{equation}
g_{k}=\frac{1}{\gamma }\left( \alpha I\pm i\alpha ^{\prime }\sigma
_{b_{k}}\right) ,
\end{equation}%
with the sign determined by the number of $\sigma _{Y}$'s in $\sigma _{b_{k}}
$. Then, by defining

\begin{equation}
\tilde{\sigma}_{b_{k}}=(-i)^{|b|_{Y}}\sigma _{b_{k}},
\end{equation}%
where $|b|_{Y}$ is the (always odd) number of $\sigma _{Y}$'s occurring in $%
\sigma _{b_{k}}$, we may write 
\begin{equation}
g_{k}=\frac{1}{\gamma }(\alpha I+\alpha ^{\prime }\tilde{\sigma}_{b_{k}}),
\label{gates}
\end{equation}%
which is the desired representation of real-valued gates \cite{Laflamme}.

Now define $C$ to be the block diagonal matrix whose blocks consist of

\begin{equation}
c=\left( 
\begin{array}{cc}
0 & 1 \\ 
0 & 0%
\end{array}%
\right) ,
\end{equation}%
i.e., 
\begin{equation}
C=\bigoplus_{i=1}^{n}c.  \label{eq:C}
\end{equation}%
Then the property that $b$ has an odd number of $11$'s is given by $b^{t}Cb=1
$. In addition, we have the multiplication rule

\begin{equation}
\tilde{\sigma}_{b_{1}}\tilde{\sigma}_{b_{2}}=(-1)^{b_{1}^{t}Cb_{2}}\tilde{%
\sigma}_{b_{1}\oplus b_{2}}  \label{eq:brule}
\end{equation}%
where the addition in the subscript is bit by bit modulo $2$.

Let $H$ be the $(2n\times N)$ matrix whose columns are the $b_{k}$:%
\begin{equation}
H=(b_{1}\text{ }b_{2}\text{ }\cdots \text{ }b_{N}).  \label{eq:H}
\end{equation}%
$H$ is a linear size, bijective representation of the quantum circuit, where
each column represents a gate and every pair of rows represents a qubit.

\begin{mydefinition}
A matrix $H$ which is constructed according to Eq.~(\ref{eq:H}) is called
the \textquotedblleft $H$-matrix representation\textquotedblright\ of the
quantum circuit $\Omega $.
\end{mydefinition}

Using the rule (\ref{eq:brule}) we then have the following expansion \cite%
{Laflamme}:

\begin{eqnarray}
U(\Omega ) &=&\prod_{k=N}^{1}g_{k}  \notag \\
&=&\prod_{k=N}^{1}\frac{1}{\gamma }(\alpha I+\alpha ^{\prime }\tilde{\sigma}%
_{b_{k}})  \notag \\
&=&\frac{1}{\gamma ^{N}}\sum_{a}(-1)^{Q_{aa}}\alpha ^{|a|}(\alpha ^{\prime
})^{N-|a|}\tilde{\sigma}_{Ha},  \label{eq:expand}
\end{eqnarray}%
where 
\begin{equation}
Q_{aa^{\prime }}\equiv a^{t}Qa^{\prime }
\end{equation}%
and where the $N\times N$ lower-triangular matrix $Q$ is defined by%
\begin{equation}
Q\equiv \mathrm{lwtr}(H^{t}CH),  \label{eq:Q}
\end{equation}%
and where $a$ ranges over all binary column vectors of length $N$. (Thus $%
Q_{aa^{\prime }}$ is not a matrix element of $Q$; we use this notation
merely for convenience.)

In order to make contact with the partition function of the Ising model, we
shall be interested in the matrix element $\langle \mathbf{0}\lvert U(\Omega
)\rvert \mathbf{0}\rangle $, where $\rvert \mathbf{0}\rangle =\otimes
_{i=1}^{n}\rvert 0_{i}\rangle $, and $\rvert 0_{i}\rangle $ is the $+1$
eigenvector of $\sigma _{Z}^{(i)}$. For this matrix element to be non-zero
no qubit can be flipped, i.e., $U(\Omega )$ cannot contain any $\sigma _{X}$
or $\sigma _{Y}$ factors. When taking the same matrix element of the
right-hand side of Eq.~(\ref{eq:expand}) we have $\langle \mathbf{0}\lvert 
\tilde{\sigma}_{Ha}\rvert \mathbf{0}\rangle $, and similarly, for this to be
non-zero $\tilde{\sigma}_{Ha}$ cannot have $\sigma _{X}$ or $\sigma _{Y}$
factors. This is enforced by summing only over those binary vectors $a$ such
that $CHa=0$ \cite{Laflamme}. Thus: 
\begin{equation}
\langle \mathbf{0}\lvert U(\Omega )\rvert \mathbf{0}\rangle =\frac{1}{\gamma
^{N}}\sum_{a\in \ker CH}(-1)^{Q_{aa}}\alpha ^{|a|}\alpha ^{\prime N-|a|}.
\label{eq:0U0}
\end{equation}%
A glance at the QWGT expression (\ref{eq:S})\ reveals a striking similarity
to the latter matrix element.

\subsection{Example}

As a simple example meant to illustrate the correspondence between the $H$%
-matrix representation of a quantum circuit $\Omega $ and the actual
operation of the circuit, consider

\begin{equation*}
H=\left[ 
\begin{array}{ccc}
1 & 1 & 1 \\ 
0 & 0 & 1 \\ 
0 & 1 & 1 \\ 
1 & 0 & 0 \\ 
1 & 1 & 1 \\ 
1 & 1 & 0%
\end{array}%
\right]
\end{equation*}%
This matrix represents a circuit $\Omega $ comprising three gates (three
columns) acting on three qubits (two rows per qubit),%
\begin{equation*}
U(\Omega )=g_{3}g_{2}g_{1},
\end{equation*}%
with the following unitaries: 
\begin{eqnarray*}
g_{1} &=&e^{-i\frac{\theta }{2}\sigma _{Z}^{(1)}\sigma _{X}^{(2)}\sigma
_{Y}^{(3)}}=\frac{1}{\gamma }(\alpha I-i\alpha ^{\prime }\sigma
_{Z}^{(1)}\sigma _{X}^{(2)}\sigma _{Y}^{(3)}), \\
g_{2} &=&e^{-i\frac{\theta }{2}\sigma _{Z}^{(1)}\sigma _{Z}^{(2)}\sigma
_{Y}^{(3)}}=\frac{1}{\gamma }(\alpha I-i\alpha ^{\prime }\sigma
_{Z}^{(1)}\sigma _{Z}^{(2)}\sigma _{Y}^{(3)}), \\
g_{3} &=&e^{-i\frac{\theta }{2}\sigma _{Y}^{(1)}\sigma _{Z}^{(2)}\sigma
_{Z}^{(3)}}=\frac{1}{\gamma }(\alpha I-i\alpha ^{\prime }\sigma
_{Y}^{(1)}\sigma _{Z}^{(2)}\sigma _{Z}^{(3)}).
\end{eqnarray*}%
The superscripts represent which qubit is being acted upon and we have
omitted the tensor product symbols. The Pauli operators can be read off from
the corresponding column entries in $H$; thus the entry $(1$ $0)^{t}$ in the
top position of the first column of $H$ represents the $\sigma _{Z}^{(1)}$
Pauli matrix in $g_{1}$, etc.

\subsection{QWGTs and the Ising Partition Function}

In Ref. \cite{Lidar:04} it was shown that the Ising partition function can
be expressed in terms of a QWGT. Let $A$ be the incidence matrix of the
graph $G=(E,V)$, i.e., 
\begin{equation}
A_{v,(i,j)}=\left\{ 
\begin{array}{ll}
1 & \mbox{$(v=i \:\: {\rm and}\:\: (i,j)\in E)$} \\ 
0 & \mbox{${\rm else}$}%
\end{array}%
\right. .
\end{equation}%
Let us associate a binary vector 
\begin{equation}
w=(w_{12},w_{13},\dots )
\end{equation}%
of length $|E|$ with the bond distribution $\{J_{ij}=\pm J\}$, by letting

\begin{equation}
w_{ij}={\frac{1-J_{ij}/J}{2}},
\end{equation}%
so that $w$ specifies whether edge $(i,j)$ supports a ferromagnetic ($%
w_{ij}=0$) or antiferromagnetic ($w_{ij}=1$) bond. Thus we can give an
equivalent definition of an instance of the Ising model (recall Definition %
\ref{def:Ising}) as the data $\Delta \equiv (G,w)$.

Let 
\begin{equation}
\lambda =\tanh (\beta J),  \label{eq:l}
\end{equation}%
and define the $|E|\times |E|$ matrix 
\begin{equation}
B=\mathrm{dg}(w)=\left\{ 
\begin{tabular}{ll}
$w$ & on the diagonal \\ 
$0$ & elsewhere%
\end{tabular}%
\right. .
\end{equation}%
Writing the instance data as $\Delta \equiv (G,w)$ we then have (Theorem 2
of \cite{Lidar:04}): 
\begin{eqnarray}
Z_{\Delta }(\lambda ) &=&\frac{2^{|V|}}{(1-\lambda ^{2})^{|E|/2}}\sum_{a\in
\ker A}(-1)^{a^{t}Ba}\lambda ^{|a|}  \notag \\
&=&\frac{2^{|V|}}{(1-\lambda ^{2})^{|E|/2}}S(A,\mathrm{dg}(w),\lambda ,1),
\label{Z} \\
&=&\frac{2^{|V|}}{(1-\lambda ^{2})^{|E|/2}}\sum_{a\in \ker A}(-1)^{a\cdot
w}\lambda ^{|a|}  \label{eq:Z1}
\end{eqnarray}%
where $a$ in the sums ranges over all $0-1$ vectors of length $|E|$
satisfying $Aa=0$, where $a^{t}Ba=\sum_{i}a_{i}w_{i}a_{i}=a\cdot w$ (since $%
a_{i}=0$ or $1$) was used in the second equality, and where the QWGT\
definition (\ref{eq:S}) was used in the last equality.

This establishes the link between QWGTs and the Ising model partition
function. Because of the similarity to the matrix element $\langle \mathbf{0}%
\lvert U(\Omega )\rvert \mathbf{0}\rangle $, we expect to be able to relate
the partition function to quantum circuits, via QWGTs. We take this up in
the next section.

\begin{mydefinition}
An even subgraph of a graph $G$ (or equivalently an Eulerian subgraph) is
any subgraph of $G$ whose vertices are of even degree. Equivalently, these
are paths in $G$ which begin and end at the same vertex, and which pass
through each edge exactly once. \label{def:even}
\end{mydefinition}

Now, note that the sum in 
\begin{equation}
S(A,\mathrm{dg}(w),\lambda ,1)=\sum_{a\in \ker A}(-1)^{a\cdot w}\lambda
^{|a|}
\end{equation}%
is over vectors that are in the kernel (nullspace) of $A$, which here means
that only subgraphs having an even number of bonds emanating from all
vertices are allowed, i.e., the sum is taken over all even subgraphs or
equivalently, all Eulerian subgraphs.

In this work we will sometimes refer to even or Eulerian subgraphs as cycles.

\section{Connecting the Ising Model Partition Function to the Quantum
Circuit Matrix Element}

\label{sec:mapping}

Our goal in this section is to connect the partition function $Z$ to $%
\langle \mathbf{0}\lvert U(\Omega )\rvert \mathbf{0}\rangle $. To do so we
will use a mapping found and described in detail in Ref. \cite{JOE2}.

\subsection{A circuit ansatz}

\label{sec:ansatz}

Focusing on the representation of the quantum circuit given in Eq.~(\ref%
{eq:expand}) and of the partition function given in Eq.~(\ref{eq:Z1}), we
begin by asking ourselves if there exists some ansatz for the gate set $%
g_{k} $ such that

\begin{equation}
U(\Omega )=\prod_{k=N}^{1}g_{k}\overset{?}{\propto }\sum_{a}(-1)^{a\cdot
w}\lambda ^{|a|}\tilde{\sigma}_{Ha},  \label{eq:U1}
\end{equation}%
where $\lambda =\tanh (\beta J)$. If such a form were possible, the
two representations would be closely linked. Indeed, we can almost get this
form. Let us take as an ansatz%
\begin{equation}
g_{k}=\frac{1}{\sqrt{\lambda ^{2}+1}}(\lambda I+\tilde{\sigma}_{b_{k}}),
\label{ansatz}
\end{equation}%
i.e., the special case of Eq.~(\ref{gates}) with $\lambda =\alpha
^{\prime }/\alpha $, or 
\begin{equation}
\tanh (\beta J)=\tan (\theta /2).  \label{eq:abl}
\end{equation}%
Note that since the inverse temperature $\beta $ and the bond strength $J$
are both positive, Eq.~(\ref{eq:abl}) restricts$\ (\theta /2)\mathrm{mod}%
2\pi $ to be in the range $(0,\pi /2)\cup (\pi ,3\pi /2)$, or $\theta 
\mathrm{mod}4\pi $ to be in the range   
\begin{equation}
R\equiv (0,\pi )\cup (2\pi ,3\pi )  \label{eq:range}
\end{equation}%
(the range for which $\tan (\theta /2)>0$). Fortunately, this includes the
case $\theta =2\arcsin (4/5)\approx 1.85\in R$ (i.e., $\lambda =4/3$),
which, as noted above, allows a universal set of one and two qubit gates to
be obtained.\footnote{%
These observations were used in Ref. \cite{JOE2} to show that finding
additive approximations of the signed generating function of Eulerian
subgraphs over hypergraphs is BQP-complete.} Thus, we have not restricted
the generality of the class of quantum circuits so far. On the other hand,
most $\theta \in R$ do not correspond to universal quantum circuits.

Next, we obtain from Eq.~(\ref{eq:expand}):

\begin{eqnarray}
U(\Omega ) &=&\prod_{k=N}^{1}\frac{1}{\sqrt{\lambda ^{2}+1}}(\lambda I+%
\tilde{\sigma}_{b_{k}})  \notag \\
&=&\frac{1}{(\lambda ^{2}+1)^{N/2}}\sum_{a}(-1)^{Q_{aa}}\lambda ^{|a|}\tilde{%
\sigma}_{Ha}.  \label{U}
\end{eqnarray}%
After taking matrix elements $\langle \mathbf{0}\lvert \cdot \rvert \mathbf{0%
}\rangle $, we have, recalling Eq.~(\ref{eq:0U0}): 
\begin{equation}
\langle \mathbf{0}\lvert U(\Omega )\rvert \mathbf{0}\rangle =\frac{1}{%
(\lambda ^{2}+1)^{N/2}}\sum_{a\in \ker CH}(-1)^{Q_{aa}}\lambda ^{|a|}.
\label{0U0}
\end{equation}%
Comparing Eqs.~(\ref{eq:Z1})\ and (\ref{0U0}), while using $\lambda
=\alpha ^{\prime }/\alpha $, we see that a sufficient condition for them
to be equal, is to identify the incidence matrix $A$ with $CH$ via 
\begin{equation}
\tilde{A}=CH,  \label{A=CH}
\end{equation}%
where $\tilde{A}$ is a matrix containing twice the number of rows of the
incidence matrix $A$, but where each even row is a zero row, and each
consecutive odd row is equal to a row of the original incidence matrix $A$,
and to equate the exponents, i.e., find an edge distribution $w$ which solves%
\begin{equation}
a\cdot w\,\mathrm{mod}\,2=Q_{aa}\quad \forall a\in \ker A  \label{eq:w}
\end{equation}%
where $Q=\mathrm{lwtr}(H^{t}\tilde{A})$ [recall Eq.~(\ref{eq:Q})]. For then

\begin{eqnarray}
\langle \mathbf{0}\lvert U(\Omega )\rvert \mathbf{0}\rangle &=&\frac{1}{%
(\lambda ^{2}+1)^{N/2}}\sum_{a\in \ker A}(-1)^{a\cdot w}\lambda ^{|a|} 
\notag \\
&=&\frac{(1-\lambda ^{2})^{\frac{|E|}{2}}}{(1+\lambda ^{2})^{\frac{|E|}{2}%
}2^{|V|}}Z_{\Delta }(\lambda ).  \label{eq:element}
\end{eqnarray}

\emph{Equation }(\ref{eq:element})\emph{\ is a key result of this paper, as
it establishes the equivalence between quantum circuits and the Ising model,
for bond distributions }$w$\emph{\ that satisfy Eq. }(\ref{eq:w})\emph{, and 
}$\lambda $\emph{'s that satisfy Eq. }(\ref{eq:abl})\emph{. }

It has two consequences. First, if we are able to determine $\langle \mathbf{%
0}\lvert U(\Omega )\rvert \mathbf{0}\rangle $, then we are able to determine
the partition function $Z_{\Delta }(\lambda )$. Note that estimating $%
\langle \mathbf{0}\lvert U(\Omega )\rvert \mathbf{0}\rangle $ in general is
BQP-complete \cite{jordan} and thus something one could do with a universal
quantum computer. Alternatively, if we had a way of classically computing $%
Z_{\Delta }(\lambda )$, then we would be able to classically simulate the
quantum circuit $\Omega $ (if it were solving a decision problem) \cite{z2}.
This latter alternative is the one we focus on in this paper.

\subsection{Circuit-Ising model compatibility}

The connections we established in the previous subsection between quantum
circuits and the partition function imply certain restrictions. We flesh
these out in the present subsection.

First, since we wish to work only with the physically relevant range of
positive temperatures and positive $J$, we restrict the gate angles $\theta $
from now on to lie in $R$. Formally:

\begin{mydefinition}
A gate angle $\theta $ [Eq.~(\ref{eq:theta})] for a gate $g_{k}$ [Eq.~(\ref%
{gates})] is said to be \textquotedblleft $\lambda $-compatible%
\textquotedblright\ if\ $\theta \in R$, where the range $R$ is defined in
Eq.~(\ref{eq:range}).
\end{mydefinition}

Next, we note that Eq.~(\ref{A=CH}) gives rise to a compatibility relation
between circuits and graphs:

\begin{mydefinition}
A quantum circuit $\Omega $, constructed with $\lambda $-compatible angles,
is \textquotedblleft $G$-compatible\textquotedblright\ with a graph $G$ if
the $H$-matrix representation of $\Omega $ satisfies Eq.~(\ref{A=CH}), where 
$A$ is the incidence matrix of $G$, and $C$ is defined in Eq.~(\ref{eq:C}).
\end{mydefinition}

When we take a $G$-compatible circuit and plug its $H$-matrix into Eq.~(\ref%
{eq:w})\ we are not guaranteed that there exists a solution $w$. Hence we
need an appropriate restriction of the class of $G$-compatible circuits:

\begin{mydefinition}
A quantum circuit $\Omega $ is \textquotedblleft $Gw$-compatible%
\textquotedblright\ if it is $G$-compatible and if the solution set of Eq.~(%
\ref{eq:w}) is non-empty.
\end{mydefinition}

We need a similar notion for the bond distributions:

\begin{mydefinition}
A bond distribution $w$ is \textquotedblleft $G\Omega $-compatible%
\textquotedblright\ with a graph $G$ and circuit $\Omega $ if it satisfies
Eq.~(\ref{eq:w}).
\end{mydefinition}

Note that in this last definition the circuit $\Omega $ must be $Gw$%
-compatible, for otherwise we are not guaranteed that the solution set of
Eq.~(\ref{eq:w}) is non-empty. Note further that, as these definitions
imply, Eqs.~(\ref{A=CH})-(\ref{eq:element}) describe a connection between
quantum circuits and instances of the Ising model over given graphs. Namely,
any $H$ which solves Eq.~(\ref{A=CH}) is a matrix representation of a
circuit $\Omega $ which belongs to a class defined by the incidence matrix $%
A $ of a given graph $G$. In addition, we can populate the vertices of $G$
with weights from the bond distribution $w$ provided $w$ is compatible. Thus:

\begin{mydefinition}
\label{def:Omega_G}Let $\Gamma$ be any set of graphs for which a solution to
Eq.~(\ref{eq:w}) exists. Then $\Omega _{\Gamma w}$ is the set of circuits
which are $Gw$-compatible $\forall G \in \Gamma$.
\end{mydefinition}

\begin{mydefinition}
\label{def:I_G}$I(\Omega _{\Gamma w})$ is the class of Ising model instances 
$\{\Delta (G,w)\}_{w}$ whose graph is $G \in \Gamma$ and whose bond
distributions $\{w\}$ are $G\Omega$-compatible $\forall G \in \Gamma$.
\end{mydefinition}

Equation (\ref{eq:w}) is a system of linear equations over $GF(2)$. The
number of equations is equal to the number of even subgraphs of the given
graph or the total number of elements in the set $\ker (CH)$, and the number
of unknowns is equal to the number of edges. However, in spite of the
fact that the number of elements in $\ker (CH)$ scales exponentially in the
number of vertices, it turns out that finding a $w$ which solves Eq. (\ref%
{eq:w}) can be done efficiently [see Eq. (\ref{eq:w-sol}) below]. Let us
further stress\emph{\ }that Eq.~(\ref{eq:w}) is only a sufficient condition
for the equality of Eqs.~(\ref{eq:Z1})\ and (\ref{0U0}), and does not
capture the whole set of possible graph instances that our scheme can
handle. We define our instances via this condition because it simplifies the
analysis and it allows us to extract information about an interesting set of
quantum circuits which may be classically simulated. We discuss more general
sufficient conditions in Section \ref{nexta}, but leave the development of a
complete understanding of the actual graph instances that our mapping can
handle, and in particular finding necessary conditions for the equality of
Eqs.~(\ref{eq:Z1})\ and (\ref{0U0}), as a problem for future study.

\section{Circuits Corresponding to Certain Planar Graphs have an Efficient
Classical Simulation}

\label{sec:proof}

Let us recap the general idea we have developed so far. At the basis of our
construction are an inverse temperature $\beta $, bond strength $J$, and a
given graph $G$. We use this graph to first identify a compatible class of
quantum circuits $\Omega _{G}$ (Definition \ref{def:Omega_G}). This class is
restricted to a subclass $\Omega _{Gw}$ of circuits for which there exist
solutions to Eq.~(\ref{eq:w}). Such solutions are used to assign weights to
the graph's edges (a bond distribution), which yields a class of Ising model
instances compatible with $G$ and $\Omega _{G}$ (Definition \ref{def:I_G}).
In other words, we go from the unweighted graph to a class of compatible
circuits, and from there to back to the graph, which is now populated by a
class of compatible Ising models. Each circuit is also parametrized by an
angle $\theta $, and when we vary $\theta $ in the range $R$ [Eq.~(\ref%
{eq:range})] we also vary over $\beta $ and $J$, via $\tan (\theta /2)=\tanh
(\beta J)$. However, not all values of $\theta $ correspond to
universal circuits. Conversely, not every circuit need correspond to a
physical (positive) temperature.

In more detail, we identify the class of quantum circuits $\Omega _{G}$
compatible with $G$ (whose incidence matrix is $A$) by solving Eq.~(\ref%
{A=CH}) for the matrices $H$ representing each (or some) $\Omega \in \Omega
_{G}$, and then find the subset $\Omega _{Gw}$ for which the solution set to
Eq.~(\ref{eq:w}) is non-empty. We then look for a bond distribution $w$ that
satisfies Eq.~(\ref{eq:w}) for a given $H$. Every such $w$ defines an Ising
model instance $\Delta (G,w)$ that is compatible with $G$ and the
corresponding $\Omega \in \Omega _{G}$. We are guaranteed that provided such
a bond distribution $w$ exists, the partition function for the corresponding
Ising model is proportional to the matrix element $\langle \mathbf{0}\lvert
U(\Omega )\rvert \mathbf{0}\rangle $ [Eq.~(\ref{eq:element})]. In other
words, any bond distribution $w$ that satisfies Eq.~(\ref{eq:w}) induces a
direct connection between quantum computation and the Ising model on a graph
with that same bond distribution.

It is important to emphasize that Eq.~(\ref{eq:w}) will not always have a
solution $w$. Whether or not this is the case is entirely determined by the
given graph $G$, since $G$, via its incidence matrix $A$, determines the
class of $G$-compatible circuits $\Omega _{G}$ [i.e., the matrices $H$ that
solve Eq.~(\ref{A=CH})], and together they determine $a$ and $Q_{aa}$ that
go into Eq.~(\ref{eq:w}), which $w$ needs to solve. Thus, it makes sense to
define a class of graphs for which there exists a solution $w$ to Eq.~(\ref%
{eq:w}).

\begin{mydefinition}
\label{def:Theta} $\Theta$ is the set of graphs for which a solution to Eq.~(%
\ref{eq:w}) exists.
\end{mydefinition}

In order to characterize $\Theta $ we require some basic ideas from graph
theory, such as obstruction sets and downward closure. These are reviewed in
Appendix \ref{app:graphs}. We shall prove:

\begin{mylemma}
\label{finite-obs} The obstruction set for $\Theta $ is finite.
\end{mylemma}

This will come as a consequence of $\Theta $'s downward closure and the
Robertson-Seymour theorem -- Theorem \ref{th:RS}. This is proved in Appendix %
\ref{proofs}. We shall also prove there that there are no solutions to Eq.~(%
\ref{eq:w}) for the graphs $K_{4}$ (the complete graph on four vertices) and 
$\bar{K}_{3,3}$, where $\bar{K}_{3,3}$ is $K_{3,3}$ with one edge missing
(which edge does not matter, since upon deletion of another edge one can
just relabel the edges and nodes and obtain exactly the same incidence
structure). Formally:

\begin{mylemma}
\label{lem:obs} The obstruction set for $\Theta $ includes $\bar{K}_{3,3}$
and $K_{4}$.
\end{mylemma}

The proof is described in Appendix \ref{proofs}. See Fig.~\ref{graphs} in
Appendix \ref{app:graphs} for a pictorial representation of the graphs
mentioned in Lemma~\ref{lem:obs}. As a consequence we will find that \emph{%
all graphs in }$\Theta $ \emph{are planar}.

We shall clarify this conclusion and the previous lemma in the next
subsection, but given their validity, a $G\Omega $-compatible bond
distribution $w$ necessarily corresponds to a \emph{planar} graph $G$. From
this it follows that its partition function can be efficiently computed
classically, and hence the corresponding class of quantum circuits, i.e., $%
\Omega _{\Theta w}$, also has an efficient classical simulation. Let us be
precise about what we mean by \textquotedblleft classically efficiently
simulatable\textquotedblright\ (CES).

\begin{mydefinition}
\label{def:CES}A uniform family $\mathcal{G}_{n}=\{\Omega _{i}\}$ of $n$%
-qubit quantum circuits is \textquotedblleft classically efficiently
simulatable\textquotedblright\ (CES)\ if the matrix element $|\langle 
\mathbf{0}\lvert U(\Omega _{i})\rvert \mathbf{0}\rangle |$ of each circuit
in $\mathcal{G}_{n}$ can be obtained to $k$ digits of precision in time $%
\mathrm{poly}(n,k)$ by classical means \cite{z2}.
\end{mydefinition}

This definition is a modified version of the one given in Ref. \cite%
{jozsa-2008}, which also includes a discussion on how it can be weakened.

When we collect the observations above we arrive at an efficient classical
test for whether a given quantum circuit is CES. This is summarized in
Theorem \ref{th}, which is our main result and the subject of the remainder
of the paper:

\begin{mytheorem}
\textbf{\emph{(Circuits Corresponding to Certain Planar Graphs have an
Efficient Classical Simulation)}}\newline
\label{th} The class of quantum circuits $\Omega _{\Theta w}$ is CES.
Deciding whether a given a graph $G$ is in $\Theta $ can be efficiently
decided.
\end{mytheorem}

The theorem comprises two parts. In the first we characterize an entire
class of CES\ quantum circuits. The proof we offer below is not
constructive, i.e., we prove that there \emph{exists} an efficient classical
simulation of the class of quantum circuits $\Omega _{\Theta w}$, and also
provide a test of non-membership in $\Omega _{\Theta w}$ for a given quantum
circuit. In the second part we given an explicit construction which decides
whether a given graph belongs to the set of graphs resulting in CES
circuits. To illustrate this part, we discuss a class of graphs (which is a
subset of $\Omega _{\Theta w}$) in Section \ref{sparsegraphs}, for which we
can explicitly find the $G\Omega $-compatible bond distribution. This
class is highly restricted in that the number of even subgraphs only grows
polynomially in the number of vertices, whereas in general, including
restrictions to planar graphs, the number of even subgraphs grows
exponentially. Nonetheless, the class of quantum circuits that one obtains
under this restriction is interesting in light of some new results about the
classical simulatability of quantum circuits \cite{Nest:08,jozsa-2008}.

For the benefit of the reader we summarize the scheme of the first
claim of the proof informally. This will also serve to summarize again
the mapping between quantum circuits and graphs.

\begin{enumerate}
\item \textbf{Given:} any subset $\Gamma $ of $\Theta $. Every $G\in \Gamma $
has a $G\Omega $-compatible bond distribution for some quantum circuit $%
\Omega $, by assumption.

\item Take the incidence matrices $CH$ of the graphs in $\Gamma$ and
transform them into the \emph{H-matrix representations} of the corresponding
quantum circuits. The following constraint must be respected: Every column
must have one $Y$-operation and can have at most one $X$-operation. [This
constraint comes from the fact that $CH$ should be an incidence matrix for a
graph, where $C$ is the block-diagonal matrix defined in Eq.~(\ref{eq:C}).
Without it one has a correspondence between quantum circuits and hypergraphs 
\cite{JOE2}. Indeed, if the incidence matrix has more than two ones per
column than one has a hypergraph.]

\item Thus $\Gamma$ corresponds to a set of quantum circuits $\Omega_{\Gamma
w}$, i.e., every quantum circuit $\Omega \in \Omega_{\Gamma w}$ is \emph{%
Gw-compatible} for some $G \in \Gamma$.

\item Show that our mapping from circuits to graphs defines a
\textquotedblleft downward closed set\textquotedblright\ of graphs, which
means that we may apply the Robertson-Seymour Theorem \cite{RS}. This
theorem guarantees that there is a finite set of graphs (obstruction set)
for which we can test whether or not $G$ has any members of this set as a
graph minor \cite{Graph} (at most cubic complexity in the number of quantum
gates in $\Omega$).

\item Define $\Theta $, via this obstruction set, i.e., a graph is a member
of $\Theta $ if it does not have $K_{4}$ and $\bar{K}_{3,3}=K_{3,3}$ \emph{%
with one edge deleted} as minors (there may be other forbidden minors). One
has a set of circuits which correspond to the graphs which have a satisfying
bond distribution $w$ for equation (\ref{eq:w}). We call the corresponding
class of quantum circuits $\Omega _{\Theta w}$. (This specific obstruction
set has been tested with mathematical software. See Appendix \ref{app:algo}.)

\item Due to the fact that these graphs are planar, the partition function $%
Z $ of any graph in $\Theta$ can be computed efficiently by a classical
computer \cite{Welsh:book}.

\item Using equation (\ref{eq:element}), show that knowledge of $Z$ can be
used to determine the outcome of a quantum circuit $\Omega \in
\Omega_{\Gamma w}$ for a decision problem.

\item \textbf{Conclude:} Families of quantum circuits in $\Omega _{\Theta w}$
which solve a decision problem can be classically simulated.
\end{enumerate}

Being that $\Gamma $ is a subset of $\Theta $, any subset $\Omega _{\Gamma
w} $ of $\Omega _{\Theta w}$ is CES. \newline

Conversely, we have a test for non-membership in the set $\Omega _{\Theta w}$%
:

\begin{enumerate}
\item \textbf{Input} a quantum circuit $\Omega$.

\item Transform $\Omega $ into a matrix whose columns represent Pauli
operations (that are to be exponentiated) and every pair of rows are the
qubits being acted upon as described in Section \ref{sec:QWGT-QC}. This
matrix is called $H$ [Eq.~(\ref{eq:H})] and is in 1-to-1 correspondence with 
$\Omega $. As above, the following constraint must be respected: Every
column must have one $Y$-operation and can have at most one $X$-operation.

\item After the above transformation, construct a corresponding incidence
matrix $CH$ of a graph $G$.

\item Check the graph $G$ for the minors $\bar{K}_{3,3}$ and $K_4$.

\item \textbf{Conclude:} If either of these are minors of $G$ then reject $%
\Omega$.
\end{enumerate}

\subsection{Ordering Lemma}

The following lemma allows us to introduce an ordering on the elements in $%
\Theta$.

\begin{mylemma}
\label{graphK} If a graph $G$ is a member of $\Theta$, then so is $%
G\setminus e_{j}$ or $G/e_{j}$, i.e., the deletion or contraction of an
arbitrary edge $e_{j}$ from a graph in $\Theta$ is also in $\Theta$.
\end{mylemma}

The proof of this lemma is technical and is given in Appendix \ref{proofs}.
Lemma \ref{graphK} implies that $\Theta $ is a \emph{downwardly closed set
with respect to the minor ordering}. Hence we can apply the
Robertson-Seymour theorem \cite{RS}, Theorem \ref{th:RS}, which states that 
\emph{any graph may be tested for membership in a given downwardly closed
set of graphs by just searching the graph for a finite set of minors.} The
complexity of doing this, given knowledge of the minors one is looking for,
can be shown to be cubic in the number of edges. We implemented this to test
Eq.~(\ref{eq:w}) for non-planar solutions, as we describe next.

\subsection{Equation (\protect\ref{eq:w}) implies planarity}

Using mathematical software we demonstrated that the mapping between graphs
and circuits described above, with the sufficient condition given by Eq.~(%
\ref{eq:w}), cannot be satisfied for $K_{5}$ and $K_{3,3}$. That is, $K_{5}$
and $K_{3,3}$ are forbidden minors for $\Theta $. The algorithm we
implemented to check this is described in Appendix \ref{app:algo}. However,
a finite graph is planar if and only if it does not have $K_{5}$ or $K_{3,3}$
as minors (Wagner's theorem; see Appendix \ref{app:graphs}). As stated
earlier, we thus have

\begin{mylemma}
\label{lem:planar}All graphs in $\Theta $ are planar.
\end{mylemma}

We remark that we have been able to find examples of planar graphs for which
there do exist solutions $w$ to Eq.~(\ref{eq:w}), e.g., $K_{2,3}$. As we
explain below, this means that $\Theta $ includes planar graphs which are
not outerplanar.

\subsection{Knowledge of the matrix element determines output to a decision
problem}

The standard way in which a quantum circuit $U$ solves a decision problem,
is to measure, say, the first qubit, and decide the problem according to
this measurement outcome. In Ref. \cite{z2} it was shown that for every such
decision problem, there exists another quantum circuit $U^{\prime }$ such
that the evaluation of $\langle 0|U^{\prime }|0\rangle $ is equivalent to
the decision problem solved by applying $U$ and measuring the first qubit.
In this sense we have:

\begin{mylemma}
\label{lem:decision}Knowledge of $\langle \mathbf{0}\lvert U(\Omega )\rvert 
\mathbf{0}\rangle $ suffices to determine the output of a quantum circuit
which is being used to solve a decision problem.
\end{mylemma}

For a proof see, e.g., Ref. \cite{z2}.

\subsection{Proof of Theorem \protect\ref{th}}

Collecting everything we now prove our main theorem. We first need one
more technical Lemma:

\begin{mylemma}
\label{quad-form}\emph{A quadratic form }$x^{t}Ax$ \emph{over GF(2)\ is
linear in }$x$ \emph{(equal to }$x^{t}\mathrm{diag}(A)$\emph{) iff }$A$\emph{%
\ is symmetric.}
\end{mylemma}

Here $\mathrm{diag}(A)$ denotes a vector comprising the diagonal of $A$. The
proof of this Lemma is presented in Appendix \ref{proofs}.

\begin{proof}[Proof of Theorem \protect\ref{th}]
We start from the second claim of the Theorem, namely we prove that we
can efficiently decide whether a given graph belongs to the set $\Theta $,
and that we can find some $w$ if it does belong. Let $G$ be a given graph
and let $K$ be the matrix whose columns are a basis of $\mathrm{Ker}(A)$,
where $A$ is the incidence matrix of $G$. This means that any $a\in \mathrm{%
Ker}(A)$ may be written as $a=Kx$ where $x$ is an arbitrary $m=\mathrm{dim}(%
\mathrm{Ker}(A))$-dimensional binary vector. Using this we may rewrite Eq. (%
\ref{eq:w}) over GF(2) as 
\begin{equation}
x^{t}K^{t}QKx=(Kx)^{t}w.  \label{eq:w-new}
\end{equation}%
Since the right-hand side is linear in $x$ for all $x$, the left-hand side
must also be linear in $x$. It follows by Lemma \ref{quad-form} that $K^{t}QK
$ is symmetric, and moreover that the quadratic form $x^{t}K^{t}QKx$ can be
written as $x^{t}\mathrm{diag}(K^{t}QK)$. Thus, solving Eq. (\ref{eq:w}) for 
$w$ is equivalent to solving the linear system $x^{t}K^{t}w+x^{t}\mathrm{diag%
}(K^{t}QK)=0$, or 
\begin{equation}
x^{t}(K^{t}w+\mathrm{diag}(K^{t}QK))=0.
\end{equation}%
Since this equation must be true for all $x$, it follows that $K^{t}w+%
\mathrm{diag}(K^{t}QK)=0$ and hence that $w$ is the solution to%
\begin{equation}
K^{t}w=\mathrm{diag}(K^{t}QK).  \label{eq:w-sol}
\end{equation}%
Since $K$ is efficiently constructable, $w$ can also be found efficiently
using standard methods for solving linear equations over GF(2).

Now for the first claim of the theorem, which states that the circuits
corresponding to the graphs in $\Theta $ are CES. $\Theta $ is a downwardly
closed set of graphs which generates a set of compatible quantum circuits $%
\Omega _{\Theta w}$ via the mapping described above. In turn we have the
corresponding Ising model instances $I(\Omega _{\Theta w})$ and by
assumption, the \emph{G$\Omega $-compatible} bond distributions $w$ for the
circuits $\Omega \in \Omega _{\Theta w}$ and graphs $G\in \Theta $. Lemma %
\ref{lem:planar} states that these instances are planar. It follows that
they are CES, i.e., we may compute the Ising partition function for any of
these instances efficiently with a classical computer. This is due to a
result by Kasteleyn, who gave a classical algorithm for the exact evaluation
of the Ising partition function of any planar graph in the absence of an
external magnetic field \cite{K}. According to our definition of CES quantum
circuits (Definition \ref{def:CES}), all we need is to be able to obtain the
evaluation in a time polynomial in the number of qubits (which translates to
the number of vertices), and the desired number of bits of precision of $%
Z_{\Delta }(\lambda )$, which is achieved by the algorithm given in \cite{K}%
. Now, since $\langle \mathbf{0}\lvert U(\Omega )\rvert \mathbf{0}\rangle
\propto Z_{\Delta }(\lambda )$ [Eq.~(\ref{eq:element})], and for any graph
in $\Theta $ we have an efficient way of classically determining $Z_{\Delta
}(\lambda )$, we are thus able to determine the matrix element $\langle 
\mathbf{0}\lvert U(\Omega )\rvert \mathbf{0}\rangle $ for any quantum
circuit in $\Omega _{\Theta w}$ efficiently. It follows from Lemma \ref%
{lem:decision} that any quantum circuit in $\Omega _{\Theta w}$ which solves
a decision problem is CES.
\end{proof}

We remark that this technique can be used to prove that quantum circuits
which correspond to non-planar classes of graphs for which the Ising
partition function has efficient classical evaluation schemes, e.g., graphs
of bounded tree width, are CES. We suspect that some of the results obtained
in Ref. \cite{Markov:tensor} may be reproduced in this way.

We further remark that due to algorithms for planarity testing, given a
quantum circuit one can test if it belongs to the class $\Omega _{\Theta w}$
of CES quantum circuits. For example, a simple test follows from the
Eulerian criterion of planarity: $|E|\leq 3|V|-6$ where $|E|$ is the number
of edges and $|V|$ is the number of vertices. (This follows from the
application of a handshaking lemma to the famous relation $|F|-|E|+|V|=2$,
where $F$ is the number of faces \cite{Graph}.) Examining the close
relationship between the circuit representation $H$ and the incidence matrix 
$CH$, one can give the restriction 
\begin{equation}
\mathrm{number\hspace{2pt}of\hspace{2pt}gates}\leq 3(\mathrm{number\hspace{%
2pt}of\hspace{2pt}qubits})-6,  \label{euler}
\end{equation}%
provided that the universal gate set consists of rotations about products of
Pauli operations. A circuit $\Omega $ for which Eq.~(\ref{euler})\ holds
generates a planar graph $G$ via the mapping we have described.
Provided Eq. (\ref{eq:w}) has a non-trivial solution $w$, it follows that 
$\Omega $ is $Gw$-compatible, and hence CES by planarity of $G$.

\section{Further characterization of the class of CES\ quantum circuits $%
\Omega _{\Theta w}$}

\label{sparsegraphs} Our motivation in this subsection is to present a
result on CES quantum circuits which allows a comparison to the recent
results presented in \cite{Nest:08} and \cite{jozsa-2008}. In both papers,
results dependent on quantum gates being restricted to nearest neighbor
qubit operations in one dimension are presented. Via our construction we
derive a similar result, but show that the restriction to nearest-neighbor
operations and one dimension can be lifted. We begin with a simple example.

\subsection{Graphs in $\Theta $ are compatible with CES circuits which
include non-nearest-neighbor operations}

Recent work in Ref. \cite{Nest:08} demonstrates that any circuit that is
built out of $X$-rotations and nearest neighbor $Z\otimes Z$ rotations can
be efficiently simulated. Now assume that one is restricted to a class of
planar graphs, $\Theta _{p}$, for which the number of even subgraphs scales
polynomially with the number of vertices.\ This restriction is not necessary
and is introduced merely for simplicity. Let us call the corresponding
set of quantum circuits (under the mapping presented above) $\Omega _{p}$.
Upon inspection of the incidence matrix of a typical graph in $\Theta _{p}$
one sees that even though the majority of incident vertices are nearest
neighbor, there are several that are not, no matter how one labels the
vertices. For example, consider the graph depicted in Fig. \ref{pc}.

\begin{figure}[tph]
\centering
\includegraphics[scale=0.5]{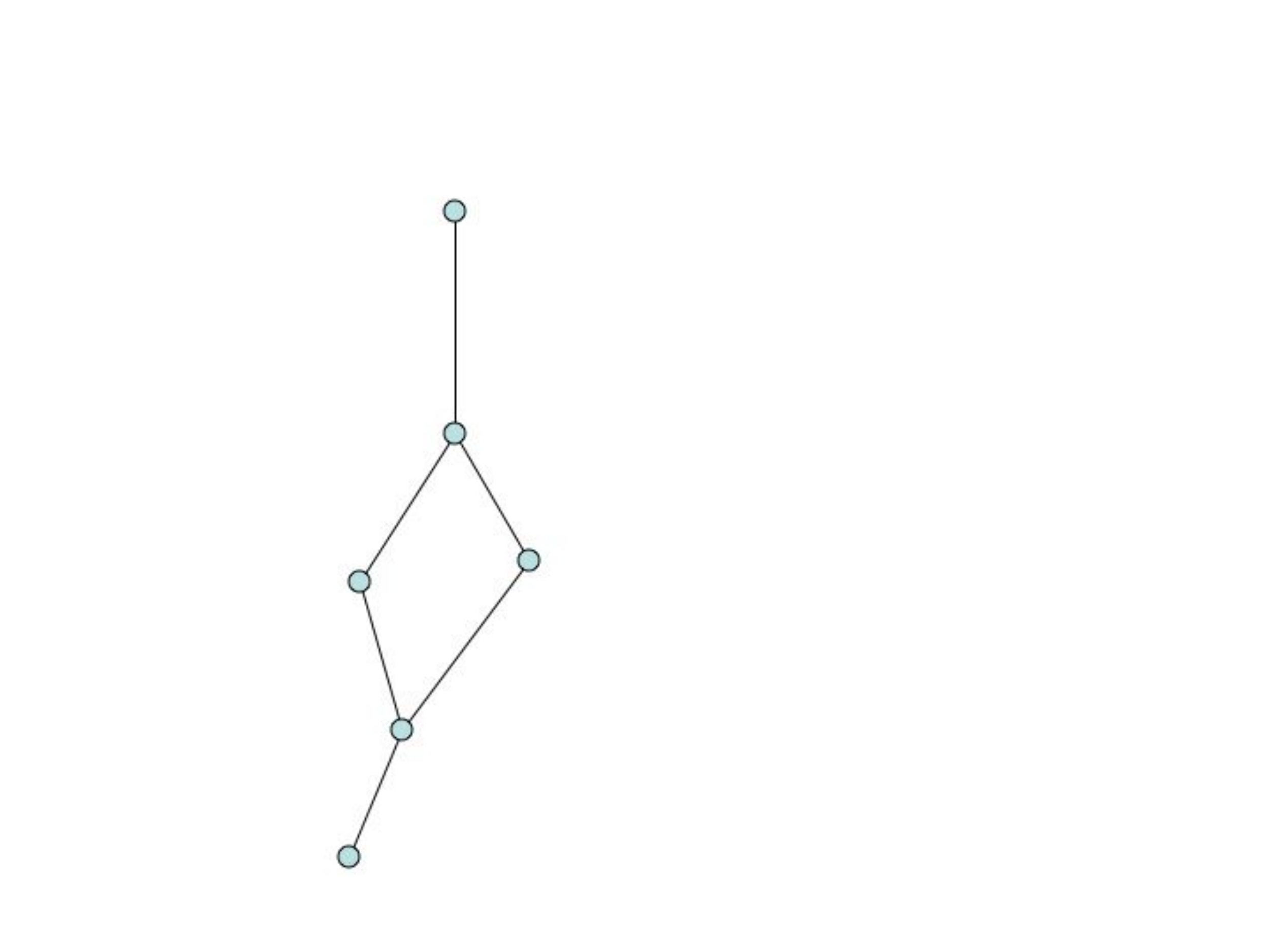}
\caption{A graph with only one cycle.}
\label{pc}
\end{figure}

The incidence matrix is given by 
\begin{equation*}
\left [ 
\begin{array}{cccccc}
1 & 0 & 0 & 0 & 0 & 0 \\ 
1 & 1 & 0 & 0 & 0 & 1 \\ 
0 & 1 & 1 & 0 & 0 & 0 \\ 
0 & 0 & 1 & 1 & 1 & 0 \\ 
0 & 0 & 0 & 1 & 0 & 0 \\ 
0 & 0 & 0 & 0 & 1 & 1%
\end{array}
\right ]
\end{equation*}
and one possible circuit representation $H$ is given by 
\begin{equation*}
\left [ 
\begin{array}{cccccc}
1 & 0 & 0 & 0 & 0 & 0 \\ 
1 & 0 & 0 & 0 & 0 & 0 \\ 
0 & 0 & 0 & 0 & 0 & 1 \\ 
1 & 1 & 0 & 0 & 0 & 1 \\ 
0 & 1 & 1 & 0 & 0 & 0 \\ 
0 & 1 & 1 & 0 & 0 & 0 \\ 
0 & 0 & 0 & 1 & 0 & 0 \\ 
0 & 0 & 1 & 1 & 1 & 0 \\ 
0 & 0 & 0 & 0 & 0 & 0 \\ 
0 & 0 & 0 & 1 & 0 & 0 \\ 
0 & 0 & 0 & 0 & 1 & 0 \\ 
0 & 0 & 0 & 0 & 1 & 1%
\end{array}
\right ]
\end{equation*}

Note that the fifth and sixth columns correspond to gates of the form 
\begin{equation*}
e^{-i\theta (\sigma _{X}^{(4)}\otimes \sigma _{Y}^{(6)})}\hspace{0.6cm}%
\mathrm{and}\hspace{0.6cm}e^{-i\theta (\sigma _{Y}^{(2)}\otimes \sigma
_{X}^{(6)})}
\end{equation*}%
respectively. The superscripts indicate which qubit is being operated on,
and thus one can clearly see that \emph{non-nearest neighbor interactions
are possible}. This example demonstrates that our construction may extend
the results in \cite{Nest:08}, for example, by linking together graphs like
the \textquotedblleft necklace\textquotedblright\ shown. Note, however,
that the nearest-neighbor restriction can only be relieved slightly as most
of the interactions will in fact remain nearest-neighbor. Taking this
example as motivation, what follows is a more\ general construction which
will be used to pursue a better understanding of CES circuits.

\subsection{A class of CES\ circuits with non-nearest neighbor gates}

We now define a simple subclass of planar graphs which have a polynomial (in
the number of vertices) number of even subgraphs. Because of planarity
this class of graphs correspond to CES quantum circuits. However, note that
this class of graphs is by no means an exhaustive characterization of
all planar graphs which have a polynomial number of even subgraphs.

\begin{mydefinition}
A basis of the null space of the incidence matrix $CH$ is referred to as a
cycle basis.
\end{mydefinition}

\begin{mydefinition}
Let $\Theta _{pc}$ be those planar graphs with $V$ vertices and $E=V+O(\log
V^{k})$ edges, where $k\in \mathbb{R^{+}}$.
\end{mydefinition}

\begin{myproposition}
\label{polyinV} $\Theta _{pc}$ has a polynomial, in $V$, number of even
subgraphs.
\end{myproposition}

\begin{proof}
A cycle basis consists of a set of connected even subgraphs of a given graph
in $\Theta _{pc}$. The dimension of the null space (or the number of
elements of the cycle basis) is in this case equal to the number of edges
minus the rank of $CH$. Recall also that the rank of the incidence matrix
equals the number of edges minus the number of components (which is one in
our case). Thus, asymptotically, one has, 
\begin{equation}
\mathrm{nullity}=V+O(k\log V)-\mathrm{rank}(CH)=O(\log V^{k}).
\end{equation}%
Now, note that the null space allows all possible sums of the basis and
therefore we are left with $O(V^{k})$ elements as claimed.
\end{proof}

One could imagine graphs in $\Theta _{pc}$ as being sparse graphs consisting
of cycles (even subgraphs) strung together along trees without too many
branching points. This is due to the relationship between $E$ and $V$ given
above. One can see that a branch without a cycle always adds an additional
vertex (one edge has two vertices) and the only way that the relationship
between vertices and edges can be satisfied is if the number of branches is
kept smaller than the number of cycles. That is, there will need to be \emph{%
more cycles than edges that do not terminate at a cycle}. Further, the
incidence matrix of these structures, like the example above in Fig.~\ref{pc}%
, will have columns that consist of nearest neighbor consecutive
\textquotedblleft 1's\textquotedblright\ for the majority of positions. This
rule is broken when a tree branches and when one runs into a cycle. By
Theorem \ref{polyinV}, this can only happen $O(\log V^{k})$ times. As $E$ is
the number of gates and $V$ is the number of qubits in the corresponding
quantum circuit, we have just proven:

\begin{mycorollary}
A quantum circuit consisting of gates of the form $e^{-i\theta
(X^{(i)}\otimes Y^{(j)})}$, which act on nearest neighbor qubits except for $%
O(k\log (\#\mathrm{of}\hspace{0.08cm}\mathrm{qubits}))$ gates, which can act
on qubits $i$ and $j$ such that $|i-j|\geq 2$, is CES.
\end{mycorollary}

Note: $Z$ operations may be included in the exponent of this operator at any
position not occupied by the $X$ and $Y$ operations.

The important point in the last corollary is that we allow non-nearest
neighbor gates, thus extending the results of \cite{Nest:08}, and also \cite%
{Jozsa:08}, where only nearest-neighbor CES\ quantum circuits were
considered.

\subsection{$\Theta $ includes some but not all outerplanar graphs}

So far we have stressed the special role of $K_{3,3}$ and $K_{5}$, which led
us to the conclusion that all graphs in $\Theta $ are planar. However, it
turns out that we can be more specific, since we have also been able to show
that $K_{4}$ and $\bar{K}_{3,3}$ are forbidden minors for $\Theta $ (Lemma %
\ref{lem:obs}; see Appendix \ref{app:algo} for a description of the proof,
using mathematical software). In other words, there does not exist a
solution to Eq.~(\ref{eq:w}) for the graphs $K_{4}$ and $\bar{K}_{3,3}$.
These graphs play a role in characterizing the set of \emph{outerplanar
graphs} \cite{Graph}, which we define next:

\begin{mydefinition}
For any planar graph, there are regions bounded by the cycles of the graph
and an unbounded region outside of all the cycles. An outerplanar graph is a
planar graph for which every vertex is within the unbounded region when it
is embedded in the plane such that no edges intersect.
\end{mydefinition}

For example, the graph in Fig.~\ref{pc} is outerplanar. More informally, a
graph is outerplanar if it can be embedded in the plane such that all
vertices lie on the outer (exterior) face. A graph $G$ is outerplanar iff $%
K_{1}+G$ (a new vertex is connected to all vertices of $G$) is planar \cite%
{Wiegers:06}. The characterization of relevance to us is the following
analog of Kuratowski's theorem for planar graphs (described in Appendix~\ref%
{app:graphs}):

\begin{mytheorem}[Chartrand \& Harrary \protect\cite{Chartrand:67}]
\label{th:OP}A graph is outerplanar if and only if it has no subgraph
homeomorphic to $K_{4}$ or $K_{2,3}$.
\end{mytheorem}

In other words, $K_{2,3}$ is a forbidden minor for outerplanar graphs, where 
$K_{2,3}$ is like $K_{3,3}$ except that one side of the bipartite graph has
two vertices instead of three.

\begin{myproposition}
\label{prop:K33-1}If a graph is outerplanar then it does not have $\bar{K}%
_{3,3}$ as a minor.
\end{myproposition}

\begin{proof}
Assume that $\bar{K}_{3,3}$ is a minor of some $G\in \mathrm{Outerplanar}$.
Then $G$ has $K_{2,3}$ as a minor, since $K_{2,3}$ is a minor of $\bar{K}%
_{3,3}$. (This is easy to see: just contract one edge and delete another.)
But by Theorem \ref{th:OP} outerplanar graphs cannot have $K_{2,3}$ as a
minor, which is a contradiction. Thus $G$ cannot be outerplanar.
\end{proof}

Note that the converse is not necessarily true, i.e., not all graphs which
do not have $\bar{K}_{3,3}$ as a minor are outerplanar.

\begin{myproposition}
$\mathrm{Outerplanar\hspace{0.08cm}Graphs}\cap \Theta \neq \emptyset$ and $%
\mathrm{Outerplanar\hspace{0.08cm}Graphs}\neq \Theta $.
\end{myproposition}

\begin{proof}
By Lemma \ref{lem:obs}, if $G\in \Theta $ then $G$ cannot have $K_{4}$ or $%
\bar{K}_{3,3}$ as minors. By Theorem \ref{th:OP} and Proposition \ref%
{prop:K33-1}, if $G^{\prime }$ is an outerplanar graph then it cannot have $%
K_{4}$ or $\bar{K}_{3,3}$ as minors either. This suggests that $\Theta $ may
have graphs in common with the set of outerplanar graphs. We have verified,
using mathematical software, that the intersection is indeed nonempty. For
example, we have found that certain trees with cycles, which are outerplanar
by construction, are in $\Theta $. Moreover, we have verified using
mathematical software that $K_{2,3}\in \Theta $, i.e., it is not a forbidden
minor for the existence of a solution $w$. But $K_{2,3}$ is not outerplanar,
hence there are graphs in $\Theta $ which are not outerplanar (in
particular, all subdivisions of $K_{2,3}$).
\end{proof}

This is interesting since some problems are NP-complete for subclasses of
planar graphs but solvable in polynomial time for outerplanar graphs. Some
examples are the chromatic number, Hamiltonian path and Hamiltonian circuit.
Another example is the page number which is one for outerplanar graphs. This
means that we can embed the vertices on a line which divides the plane into
two subplanes, and draw all edges in one of the subplanes without crossing.
In a sense, outerplanar graphs are \textquotedblleft easy\textquotedblright\
computationally. The fact that $\Theta $ includes non-outerplanar graphs
thus suggests that it may include interesting computational problems.

\section{Discussion and Future Directions}

\label{nextstep}

In this section we briefly discuss two possible future directions for
research.

\subsection{General condition for the bond distribution}

\label{nexta} So far we have assumed the sufficient condition (\ref{eq:w})
in order to obtain the desired equality between the partition function and
the circuit matrix element $\langle \mathbf{0}\lvert U(\Omega )\rvert 
\mathbf{0}\rangle $ [Eqs.~(\ref{eq:Z1})\ and (\ref{0U0})]. We have also
shown that a satisfying bond distribution $w$ can always be efficiently
tested for an computed, under Eq. (\ref{eq:w}). Let us now relax the
constraint of Eq.~(\ref{eq:w}) by considering a more general way in which
the desired proportionality,%
\begin{equation}
\sum_{a\in \ker A}(-1)^{a\cdot w}\lambda ^{|a|}\propto \sum_{a\in \ker
CH}(-1)^{a^{t}\mathrm{lwtr}(H^{t}CH)a}\lambda ^{|a|},  \label{equal}
\end{equation}%
can be obtained. Indeed, Eq.~(\ref{eq:w}) is clearly not a necessary
condition. The following construction demonstrates that it is likely that
the number of cases which \emph{do not have a solution $w$} for the bond
distribution is much smaller than the case we analyzed given by Eq.~(\ref%
{eq:w}).

Note that in Eq.~(\ref{equal}) the powers of the $\lambda $'s are the
weights of the null vectors $a$, that is the number of ones in $a$. Thus it
is possible for an equality to occur for a given term in the sum in the two
sides of Eq.~(\ref{equal}) for different $a$'s, as long as the weights of
the $a$'s are equal. This gives us the constraint for the following. One can
organize all the $a$'s in bins in terms of weights from $1$ to $|E|=N$. Let
us now take bin $r$, i.e., the set of vectors of weight $r$. Let $%
a_{r1},\dots ,a_{rn}$ be all the null vectors of $CH$ of weight $r$. Then,
if for all $r$

\begin{eqnarray}
\{a_{rj_{1}}^{t}\mathrm{lwtr}(H^{t}CH)a_{rj_{1}} &=&a_{rj_{2}}\cdot
w\}\wedge   \notag \\
\{a_{rj_{2}}^{t}\mathrm{lwtr}(H^{t}CH)a_{rj_{2}} &=&a_{rj_{3}}\cdot
w\}\wedge \cdots \wedge   \notag \\
\{a_{rj_{n}}^{t}\mathrm{lwtr}(H^{t}CH)a_{rj_{n}} &=&a_{rj_{1}}\cdot w\},
\label{conj}
\end{eqnarray}%
where $\{j_{1},...,j_{n}\}$ is any permutation of the numbers $\{1,...,n\}$,
then Eq.~(\ref{equal}) would be satisfied. Clearly, Eq.~(\ref{eq:w}) is a
special case of this more general condition. This demonstrates that it is
likely that a satisfying bond distribution $w$ can be found for a given
graph even if the sufficient condition (\ref{eq:w}) cannot be satisfied.
Loosely, this is due the fact that there are many conjunctive statements of
the form (\ref{conj}) that can be satisfying. In fact it seems likely that
also non-planar graphs (i.e., those containing $K_{5}$ or $K_{3,3}$ as
minors) may have a satisfying $w$ according to Eq.~(\ref{conj}). We leave
this as a problem for future investigation. An important point
concerning the more general condition (\ref{conj}) is that we do not know if
it can be efficiently tested for a satisfying bond distribution $w$, let
alone solved for such a $w$.

\subsection{Computing the Ising partition function}

As mentioned in Section \ref{sec:ansatz}, Eq. (\ref{eq:element}) has two
consequences, and our focus in this paper has been on the ability to find
CES circuits using known results about the hardness of computing partition
functions. Let us now briefly consider the other consequence, namely the
fact that if we are able to determine $\langle \mathbf{0}\lvert U(\Omega
)\rvert \mathbf{0}\rangle $, then we are able to determine the partition
function $Z_{\Delta }(\lambda )$.

A fully-polynomial randomized approximation scheme (fpras) for the
fully-ferromagnetic Ising partition function was presented in \cite%
{Jerrum:90}. It is well known that having an fpras for the non-ferromagnetic
Ising model implies that $NP=RP$ (randomized polynomial time) which would be
quite unexpected \cite{Welsh:book}. It should therefore be of no surprise
that no fpras for this problem has been found, even with quantum resources.
However \emph{additive} approximation schemes seem likely and in fact one
was given in \cite{dorit-tutte} for the related Potts model partition
function, even though the instances that they were able to account for are
not known to be BQP-complete and the hardness is in fact unknown.

Equation (\ref{eq:element}) precisely relates a matrix element of a
quantum circuit with the value of the partition function of the Ising model
for a corresponding graph instance. This means that if we could
approximate the matrix element, we would have an approximation for the Ising
partition function. Due to the Hadamard test, it is well known that a
polynomial estimation of this matrix element is BQP-complete. (See Ref. \cite%
{jordan} for a description of the Hadamard test.) Specifically, by making $%
1/\epsilon ^{2}$ measurements, one can either have $\mathrm{Re}\langle 
\mathbf{0}\lvert U(\Omega )\rvert \mathbf{0}\rangle $ or $\mathrm{Im}\langle 
\mathbf{0}\lvert U(\Omega )\rvert \mathbf{0}\rangle $ to precision $\epsilon 
$, but one must keep in mind that this approximation is an additive one.
This means that with some probability of success bounded below (say by $.75$%
) the approximation returns $m$ such that 
\begin{equation}
\langle \mathbf{0}\lvert U(\Omega )\rvert \mathbf{0}\rangle -\delta \cdot
p<m<\langle \mathbf{0}\lvert U(\Omega )\rvert \mathbf{0}\rangle +\delta
\cdot p,  \label{eq:approx}
\end{equation}%
where $p$ is a polynomially small parameter and $\delta $ is the
approximation scale of the problem. Note that if $\delta =O(\langle \mathbf{0%
}\lvert U(\Omega )\rvert \mathbf{0}\rangle )$ then the approximation will be
an fpras \cite{dorit-tutte}. Equations (\ref{eq:element}) and (\ref%
{eq:approx}) taken together quantify how a measurement of $\langle \mathbf{0}%
\lvert U(\Omega )\rvert \mathbf{0}\rangle $ yields an approximation of the
partition function of the Ising model instance $\Delta $ corresponding to
the circuit $\Omega $.

\section{Conclusions}

\label{sec:conc}

We have provided a construction that allows one to determine if a given
quantum circuit corresponds to a class of quantum circuits which are
classically efficiently simulatable (CES). This was done by looking at the
corresponding graph instances of the classical Ising model using a mapping
previously introduced in \cite{JOE2}. This was then used to conclude that
any class of quantum circuits which solve decision problems and are
restricted to certain planar graph instances are CES. Our main result is
stated in Theorem \ref{th}, which characterizes the class of CES\ circuits
via the set of planar graphs $\Theta $. We have given a partial
characterization of $\Theta $ by stating that its obstruction set includes
$\bar{K}_{3,3}$ and $K_{4}$ (and hence, by downward closure also $%
K_{3,3} $ and $K_{5}$). An interesting open problem is to give a complete
characterization of the obstruction set; we know from the Robertson Seymour
Theorem that this set is finite, since we have proved that $\Theta $ is
downwardly closed.

Our mapping can also be used to construct a quantum algorithm for the
additive approximation of the partition function. However, there are two
issues. The instances we are able to handle are constrained by our use of
equation (\ref{eq:w}) which does not capture all the ways a certain bond
distribution may satisfy equation (\ref{equal}), but which simplifies our
analysis greatly. The other issue is the fact that our mapping may fail to
provide information about the bond distribution of a given graph.

An open problem is to obtain a better understanding of what the complexity
of finding the bond distribution for a particular graph instance is. This
understanding will have consequences in our knowledge of where BQP is in the
complexity hierarchy, as we will be able to relate the simulatability of
universal quantum circuit families with the complexity of finding bond
distributions. On the other hand, it is possible that the complexity of
finding bond distributions is somehow incorporated in the power of the
quantum circuit that corresponds to the graph instance of the Ising model,
in the sense that the circuit corresponding to a planar graph (under our
mapping) may not be CES, because the effort of obtaining the bond
distribution via Eq. (\ref{conj}) blocks such a simulation.

\begin{acknowledgments}
This material is based upon work supported by the National Science
Foundation (NSF) under Grant No. PHY-0802678, and by the Army Research
Office under grant W911NF-05-1-0440. DAL thanks the Institute for
Quantum Information (IQI) at Caltech where part of this work was
done. IQI is supported by the NSF under Grant No. PHY-0803371.
\end{acknowledgments}

\appendix

\label{sec:app}

\section{Essential elements from Graph Theory}

\label{app:graphs}

Here we review some essential definitions and theorems form graph theory
needed for the results presented in this work. A good reference for these
concepts is the wikipedia article on planar graphs, or Ref. \cite{Graph}.

\begin{mydefinition}
A subgraph $H$ of a given graph $G$ is called a \emph{minor} (or child) of $%
G $ if it is isomorphic to a graph that can be obtained from $G$ via a
sequence of edge deletions, edge contractions, or deletion of isolated
vertices. Edge contraction is the process of removing an edge and combining
its two endpoints into a single vertex. Edge deletion removes an edge
without removing its vertices.
\end{mydefinition}

\begin{mydefinition}
The set of graphs $S$ is \emph{downwardly closed with respect to minor
ordering} if whenever $G$ is a member of $S$, then so is any minor of $G$.
\end{mydefinition}

A trivial consequence of the definition of downwardly closed sets is that
every such set has an obstruction set:

\begin{mycorollary}
An \emph{obstruction set} for $G$ is a set of minors of $G$, also called 
\emph{forbidden minors}, with the property that they prevent downward
closure.
\end{mycorollary}

In other words, if one constructs a set of minors of $G$ and encounters a
minor that violates the property for which downward closure is being tested,
such a minor is called forbidden, and belongs to the obstruction set. Now
comes a seminal theorem due to Robertson and Seymour \cite{RS}:

\begin{mytheorem}
\label{th:RS}(Robertson--Seymour) Every downwardly closed set of graphs
(possibly infinite) has a finite obstruction set, i.e., a finite set of
forbidden minors.
\end{mytheorem}

For our purposes an important example are the planar graphs:

\begin{mydefinition}
A \emph{planar graph} is a graph which can be embedded in the plane, i.e.,
it can be drawn on the plane in such a way that its edges may intersect only
at their endpoints, i.e., edges never cross. A \emph{nonplanar graph} is a
graph which cannot be drawn in the plane without edge intersections.
\end{mydefinition}

There are two particularly important nonplanar graphs, denoted $K_{5}$ (the
complete graph on five vertices; complete means that each pair of vertices
are connected by an edge) and $K_{3,3}$ (the complete bipartite graph on six
vertices, three of which connect to each of the other three). They are
depicted in Fig.~\ref{graphs}.

\begin{figure}[tph]
\centering
\includegraphics[scale=0.4]{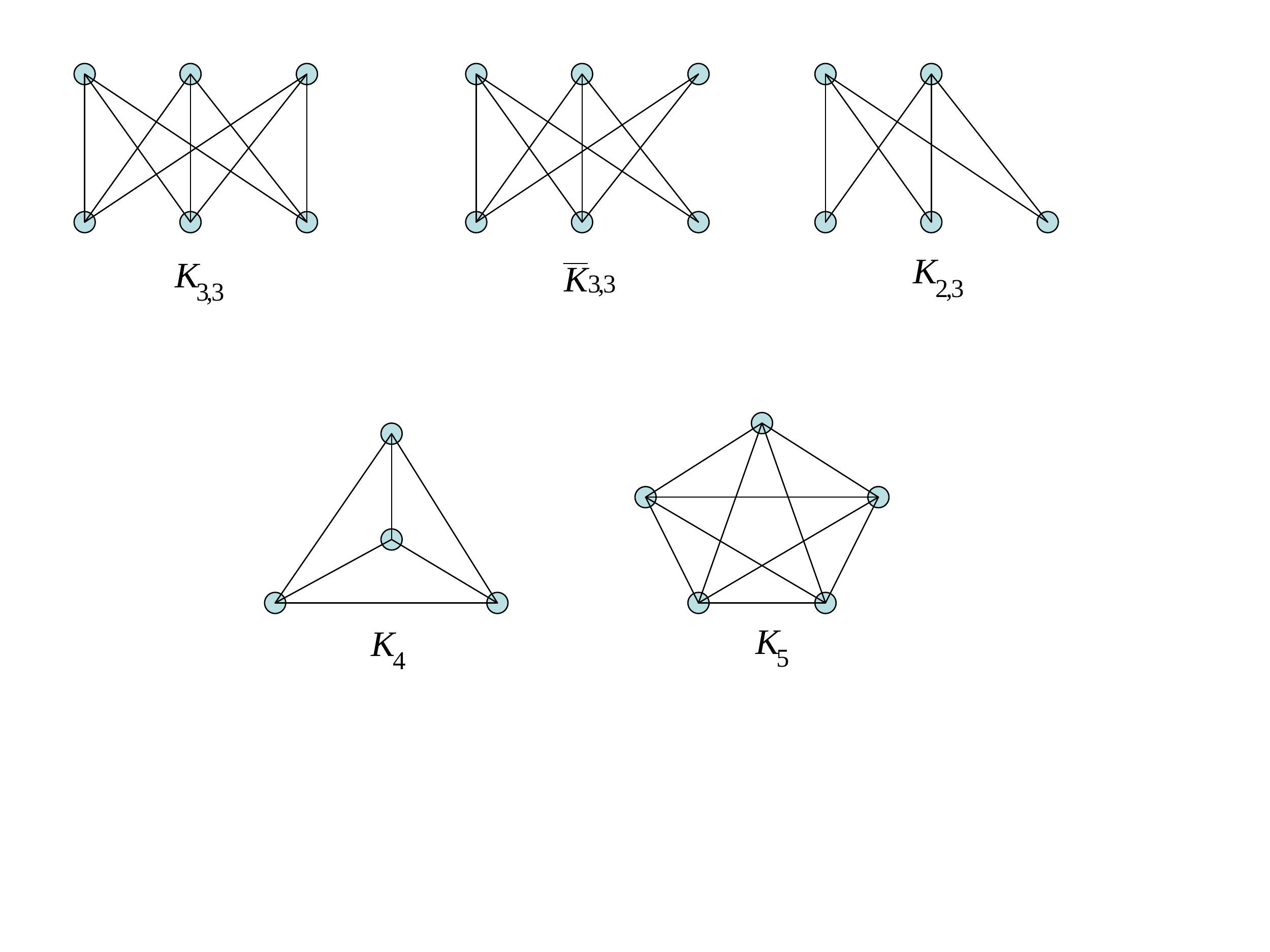} \vspace{-2.5cm}
\caption{The various graphs playing a role in defining the obstruction set
for $\Theta$ and outerplanar graphs.}
\label{graphs}
\end{figure}

Planarity is characterized by Wagner's Theorem \cite{Graph}:

\begin{mytheorem}[Wagner]
\label{th:Wag} A finite graph is planar if and only if it does not have $%
K_{5}$ or $K_{3,3}$ as a minor.
\end{mytheorem}

In other words, if one deletes or contracts the edges of a graph and finds
one of these minors, then the graph is not planar. Hence $K_{5}$ and $%
K_{3,3} $ are forbidden minors (form an obstruction set) for planarity.

For completeness we note that an alternative characterization can be given
in terms of the concept of a subdivision of a graph:

\begin{mydefinition}
A \emph{subdivision} of a graph results from inserting vertices into edges.
\end{mydefinition}

Thus, while deletion and contraction of edges shrinks a graph down,
subdivision builds it up.

\begin{mydefinition}
Two graphs $G$ and $G^{\prime }$ are \emph{homeomorphic} if there is an
isomorphism from some subdivision of $G$ to some subdivision of $G^{\prime }$%
.
\end{mydefinition}

\begin{mytheorem}[Kuratowski]
A finite graph is planar if and only if it has no subgraph homeomorphic to $%
K_{5}$ or $K_{3,3}$.
\end{mytheorem}

In other words, a finite graph is planar if and only if it does not contain
a subgraph that is isomorphic to a subdivision of $K_{5}$ or $K_{3,3}$.

In this work we give a criterion for graph membership (in a certain set $%
\Theta $ which is defined in Definition \ref{def:Theta}) based on the
existence of a solution for the system of linear equations over $GF(2)$
defined by Eq.~(\ref{eq:w}). We want to be guaranteed that this membership
has an ordering in the sense that if $G$ is a member, then so is every minor
of $G$. In other words, we are testing for downward closure. The
Robertson--Seymour theorem guarantees that membership in our set $\Theta $
is not obstructed by an \emph{infinite} set of graphs. However, the
situation is in fact far better: membership of a graph to a fixed downward
closed set can be checked by running a polynomial time algorithm for all
elements of the obstruction set (if it is known), since searching for a
minor on a given graph only requires cubic time \cite{Graph}. In fact,
checking whether a graph is planar can be done in linear time.

For completeness we include the definition of a \emph{hypergraph} as our
correspondence between quantum circuits and graphs is actually a mapping
between circuits and hypergraphs \cite{JOE2}.

\begin{mydefinition}
A hypergraph is a generalization of a graph where edges are replaced by 
\emph{hyperedges}. Let $V=\{v_{1},v_{2},\dots ,v_{k}\}$ be the set of
vertices and let $E=\{e_{1},e_{2},\dots ,e_{n}\}$ be the set of hyperedges.
Each $e_{i}=\{v_{i1},v_{i2},\dots ,v_{im}\}$ is a collection of vertices
where each $v_{ij}\in V$.
\end{mydefinition}

Thus the main difference from ordinary graphs is that edges consist of
arbitrary collections of vertices rather than two and thus graphs are
special cases of hypergraphs. As shown in Ref. \cite{JOE2}, the existence of
hyperedges is what gives us access to the universal gate set presented in 
\cite{Laflamme}, via the two assumptions enumerated at the end of Section %
\ref{sec:mapping}. However, the circuits that correspond to graphs is what
is interesting here, as our results depend on information about the Ising
partition function defined on ordinary graphs.

\section{Proof of Lemmas}

\label{proofs}

Here we present the proof of lemmas \ref{finite-obs}, \ref{graphK}, and \ref%
{quad-form}, which we repeat for convenience:

\noindent \textbf{Lemma \ref{graphK}} \emph{If a graph $G$ is a member of $%
\Theta$, then so is $G\setminus e_{j}$ or $G/e_{j}$, i.e., the deletion or
contraction of an arbitrary edge $e_{j}$ from a graph in $\Theta$ is also in 
$\Theta$.}

\begin{proof}
Assume that $G\in \Theta $. Recall that a graph $G$ is an element of $\Theta 
$ if there exists some solution $w$ to the set of linear equations over $%
GF(2)$ 
\begin{equation}
A^{(G)}w=\alpha ^{(G)}
\end{equation}%
where $A^{(G)}$ is the matrix whose rows are elements, $a_{i}$, of the
nullspace of the incidence matrix of the graph $G$ (given as the Ising
instance), and $\alpha ^{(G)}$ is the vector whose entries are the $a_{i}^{t}%
\mathrm{lwtr}(H^{t}CH)a_{i}$. These null elements, $a_{i}$, correspond to
the even subgraphs of $G$ and will be referred to as cycles (recall
Definition~\ref{def:even}). From elementary linear algebra we know that a
solution exists if $\alpha ^{(G)}$ may be written as a linear combination of
columns of $A^{(G)}$ or in other words if 
\begin{equation}
\mathrm{Rank}[A^{(G)}|\alpha ^{(G)}]=\mathrm{Rank}[A^{(G)}].
\end{equation}%
We must demonstrate that after we either delete or contract an edge, and
arrive at the subgraph $G^{\prime }$, we have 
\begin{equation}
\mathrm{Rank}[A^{(G^{\prime })}|\alpha ^{(G^{\prime })}]=\mathrm{Rank}%
[A^{(G^{\prime })}].
\end{equation}

\begin{widetext}
We shall demonstrate this with an edge deletion as the case of a contraction
is similar. We begin with a given graph $G$ which is a member of $\Gamma _{w}
$. We have 
\begin{equation}
({G^{\prime }}^{t}CG^{\prime })^{(G)}=\left[ 
\begin{array}{cccc}
G_{11}^{\prime } & G_{21}^{\prime } & \cdots  & G_{v1}^{\prime } \\ 
G_{12}^{\prime } & G_{22}^{\prime } & \cdots  & G_{v2}^{\prime } \\ 
\vdots  & \vdots  & \vdots  & \vdots  \\ 
G_{1s}^{\prime } & G_{2s}^{\prime } & \cdots  & G_{vs}^{\prime } \\ 
\vdots  & \vdots  & \vdots  & \vdots  \\ 
G_{1N}^{\prime } & G_{2N}^{\prime } & \cdots  & G_{vN}^{\prime }%
\end{array}%
\right] \left[ 
\begin{array}{cccccc}
G_{21}^{\prime } & G_{22}^{\prime } & \cdots  & G_{2s}^{\prime } & \cdots  & 
G_{2N}^{\prime } \\ 
0 & 0 & \cdots  & 0 & \cdots  & 0 \\ 
G_{41}^{\prime } & G_{42}^{\prime } & \cdots  & G_{4s}^{\prime } & \cdots  & 
G_{4N}^{\prime } \\ 
0 & 0 & \cdots  & 0 & \cdots  & 0 \\ 
\vdots  & \vdots  & \vdots  & \vdots  & \vdots  & \vdots  \\ 
G_{v1}^{\prime } & G_{v2}^{\prime } & \cdots  & G_{vs}^{\prime } & \cdots  & 
G_{vN}^{\prime } \\ 
0 & 0 & \cdots  & 0 & \cdots  & 0%
\end{array}%
\right] 
\end{equation}

Using Einstein notation this equals%
\begin{equation}
\left[ 
\begin{array}{cccccc}
G^{\prime }_{2i-1,1}{G^{\prime }}^{2i,1} & G^{\prime }_{2i-1,1}{G^{\prime }}%
^{2i,2} & \cdots & G^{\prime }_{2i-1,1}{G^{\prime }}^{2i,s} & \cdots & 
G^{\prime }_{2i-1,1}{G^{\prime }}^{2i,N} \\ 
G^{\prime }_{2i-1,2}{G^{\prime }}^{2i,1} & G^{\prime }_{2i-1,2}{G^{\prime }}%
^{2i,2} & \cdots & G^{\prime }_{2i-1,2}{G^{\prime }}^{2i,s} & \cdots & 
G^{\prime }_{2i-1,2}{G^{\prime }}^{2i,N} \\ 
\vdots & \vdots & \vdots & \vdots & \vdots & \vdots \\ 
G^{\prime }_{2i-1,N}{G^{\prime }}^{2i,1} & G^{\prime }_{2i-1,N}{G^{\prime }}%
^{2i,2} & \cdots & G^{\prime }_{2i-1,N}{G^{\prime }}^{2i,s} & \cdots & 
G^{\prime }_{2i-1,N}{G^{\prime }}^{2i,N}%
\end{array}%
\right]  \label{HCH}
\end{equation}
  
Keep in mind that if we were to delete an edge from $G$, this would
correspond to losing a column from $CH$ which would correspond to losing,
say the $s^{\mathrm{th}}$ column from matrix (\ref{HCH}). Taking the lower
triangular portion of matrix (\ref{HCH}) and calculating we find that the $%
m^{\mathrm{th}}$ element of the vector $\alpha ^{(G)}$ is%
\begin{eqnarray}
\alpha _{m}^{(G)} &=&a_{m2}[a_{m1}G_{2i-1,2}^{\prime }{G^{\prime }}%
^{2i,1}]+a_{m3}[a_{m1}G_{2i-1,3}^{\prime }{G^{\prime }}%
^{2i,1}+a_{m2}G_{2i-1,3}^{\prime }{G^{\prime }}^{2i,2}]+\cdots   \notag \\
&+&a_{mk}[a_{m1}G_{2i-1,k}^{\prime }{G^{\prime }}^{2i,1}+\cdots
+a_{m(k-1)}G_{2i-1,k}^{\prime }{G^{\prime }}^{2i,k-1}]+\cdots   \notag \\
&+&a_{mn}[a_{m1}G_{2i-1,N}^{\prime }{G^{\prime }}^{2i,1}+\cdots
+a_{m(N-1)}G_{2i-1,N}^{\prime }{G^{\prime }}^{2i,N-1}]
\end{eqnarray}%
where $a_{mi}$ is the $i^{\mathrm{th}}$ element of the $m^{\mathrm{th}}$
null vector of $CH$ or the matrix element $A_{m,i}$.

Now, let 
\begin{eqnarray}
\xi _{m}^{s} &=&a_{m(s+1)}a_{ms}G_{2i-1,s+1}^{\prime }{G^{\prime }}%
^{2i,s}+a_{m(s+2)}a_{ms}G_{2i-1,s+2}^{\prime }{G^{\prime }}^{2i,s}+\cdots  
+a_{mN}a_{ms}G_{2i-1,N}^{\prime }{G^{\prime }}^{2i,s}.
\end{eqnarray}%
\end{widetext}
This is the portion of $\alpha _{m}^{(G)}$ that would vanish if we were to
omit the edge that corresponds to the $s^{\mathrm{th}}$ column of the matrix
(\ref{HCH}). This means that if we remove this edge, we will end up with the
subgraph $G^{\prime }$ and we can write%
\begin{equation}
\alpha _{m}^{(G)}=\alpha _{m}^{({G^{\prime })}}+\xi _{m}^{s}.
\end{equation}%
This equation is saying that the $m^{\mathrm{th}}$ entry of the right hand
side of 
\begin{equation}
A^{(G)}w=\alpha ^{(G)}
\end{equation}%
is given by the $m^{\mathrm{th}}$ entry of the right hand side of 
\begin{equation}
A^{({G^{\prime })}}w=\alpha ^{({G^{\prime })}}
\end{equation}%
(the corresponding system of equations for the graph $G^{\prime }$) plus the
term $\xi _{m}^{s}$.

From the assumption that $G\in \Gamma _{w}$ and by construction we have 
\begin{equation}
\alpha ^{(G)}=\left[ 
\begin{array}{ccc}
\alpha _{1}^{({G^{\prime })}} & + & \xi _{1}^{s} \\ 
\alpha _{2}^{({G^{\prime })}} & + & \xi _{2}^{s} \\ 
\vdots  & \vdots  & \vdots  \\ 
\alpha _{K}^{({G^{\prime })}} & + & \xi _{K}^{s}%
\end{array}%
\right] =\sum_{i}\delta _{i}c_{i}
\end{equation}%
where the $\delta _{i}$ are coefficients in $GF(2)$ and the $c_{i}$ are
columns of $A^{(G)}$, i.e., $\mathrm{Rank}[A^{(G)}|\alpha ^{(G)}]=\mathrm{%
Rank}[A^{(G)}].$ Thus we have 
\begin{equation}
\alpha ^{({G^{\prime })}}=\sum_{i}\delta _{i}c_{i}-\xi ^{s}.
\end{equation}%
How does the matrix $A$ change as we go from $G\longrightarrow {G^{\prime }}$
by this edge deletion? If the edge is a \emph{dangling edge}, i.e., not part
of a cycle, then we lose a column (column $s$) but if the edge deletion
causes the breaking of $M$ cycles, then $A$ will lose $M$ rows (in addition
to column $s$), as the rows encode the cycle structure of the graph. In this
case, the dimension (or length) of $\alpha ^{(G^{\prime })}$ will be $M$
less than the dimension of $\alpha ^{(G)}$ and the $c_{i}$ will also be
shorter by $M$ entries. We call these shorter $c_{i}$, $c_{i}^{\prime }$.
Further, and most importantly, $\xi ^{s}$ will vanish, as mentioned. After
taking this into consideration we now can conclude that%
\begin{equation}
\alpha ^{(G^{\prime })}=\sum_{i\neq s}\delta _{i}c_{i}^{\prime }
\end{equation}%
where $A^{(G^{\prime })}=[c_{1}^{\prime }c_{2}^{\prime }\cdots
c_{N-1}^{\prime }].$ Thus, 
\begin{equation}
\mathrm{Rank}[A^{(G^{\prime })}|\alpha ^{(G^{\prime })}]=\mathrm{Rank}%
[A^{(G^{\prime })}].
\end{equation}

The proof for edge contractions is similar. The main difference is that an
edge contraction does not cause the loss of a cycle except when the edge in
question belongs to a cycle of length three. Thus in general, the
contraction case is simpler except when dealing with cycles of length three.
In this case the proof caries over in the same way.
\end{proof}

Thus the set of graphs $\Theta$ is downwardly closed.

\noindent \textbf{Lemma \ref{finite-obs}} \emph{The obstruction set for $%
\Theta $ is finite.}

\begin{proof}
The set of graphs $\Theta $ is \emph{downwardly closed} by the above lemma.
One may then apply the Robertson-Seymour Theorem (Theorem \ref{th:RS}) and
immediately conclude that the number of forbidden minors of $\Theta $ is
finite.
\end{proof}

\noindent \textbf{Lemma \ref{quad-form}} \emph{A quadratic form }$x^{t}Ax$ 
\emph{over GF(2)\ is linear in }$x$\emph{\ (equal to }$x^{t}\mathrm{diag}(A)$%
) \emph{iff }$A$\emph{\ is symmetric.}

\begin{proof}
Let $x$ be an $m$-dimensional column vector and $A$ an $m\times m$ matrix,
both over GF(2). Consider the quadratic form $x^{t}Ax=%
\sum_{ij}x_{i}A_{ij}x_{j}$ and assume that $A$ is symmetric:\ $A=A^{t}$.
Then 
\begin{eqnarray}
x^{t}Ax
&=&\sum_{i<j}A_{ij}x_{i}x_{j}+\sum_{i}A_{ii}x_{i}^{2}+%
\sum_{i>j}A_{ij}x_{i}x_{j}  \notag \\
&=&\sum_{i<j}A_{ij}x_{i}x_{j}+\sum_{i}A_{ii}x_{i}+\sum_{j>i}A_{ij}x_{j}x_{i},
\end{eqnarray}%
where in the second line we used $x_{i}^{2}=x_{i}$ [true over GF(2)],
exchanged $i$ and $j$ in the third summand and used $A_{ij}=A_{ji}$. The
first and third summands are equal and hence add up to zero over GF(2). We
are left with the, second, linear term, i.e., $x^{t}Ax=x^{t}\mathrm{diag}(A)$%
, where $\mathrm{diag}(A)$ denotes a vector comprising the diagonal of $A$. 

Next, assume that $A$ is not symmetric. Then there exists a pair of indices $%
i^{\prime }<j^{\prime }$ such that $A_{i^{\prime }j^{\prime }}\neq
A_{j^{\prime }i^{\prime }}$, i.e., $A_{i^{\prime }j^{\prime }}+A_{j^{\prime
}i^{\prime }}=1$. As above, we have: $x^{t}Ax=\sum_{i<j}A_{ij}x_{i}x_{j}+%
\sum_{i}A_{ii}x_{i}+\sum_{j>i}A_{ji}x_{j}x_{i}$. Consider the index pair $%
(i^{\prime },j^{\prime })$ in this sum:\ 
\begin{equation}
A_{i^{\prime }j^{\prime }}x_{i^{\prime }}x_{j^{\prime }}+A_{j^{\prime
}i^{\prime }}x_{j^{\prime }}x_{i^{\prime }}=(A_{i^{\prime }j^{\prime
}}+A_{j^{\prime }i^{\prime }})x_{i^{\prime }}x_{j^{\prime }}=x_{i^{\prime
}}x_{j^{\prime }}.
\end{equation}%
Thus the quadratic form contains at least one non-linear (quadratic) term $%
x_{i^{\prime }}x_{j^{\prime }}$.
\end{proof}

\section{Algorithm for Minor Testing}

\label{app:algo}

Note that all calculations are done modulo 2.

\textbf{Input:} A graph $G$ for which we wish to determine if there exists
some satisfying edge interaction $w$ that satisfies Eq.~(\ref{eq:w}).
Specifically, we considered $G=K_{3,3}$, $K_{5}$, $\bar{K}_{3,3}$ ($K_{3,3}$
with one edge deleted) and $K_{4}$.

\textbf{Output:} A binary vector $w$ that is a satisfying bond distribution
or a null vector (indicating no such bond distribution).

\begin{enumerate}
\item From the incidence matrix $A$ of $G$ obtain the following items:

\begin{enumerate}
\item All vectors $a_{i}$ belonging to the null space $\mathcal{L}$ of the
incidence matrix. These row vectors form a matrix $M$.

\item Construct a matrix representation $H$ of the possible corresponding
quantum circuits (under the mapping presented earlier). From $H$ construct $%
Q=\mathrm{lwtr}H^{t}CH$. This matrix will have variables $z_{i}$
corresponding to all the possible ways that one can include or omit $Z$
operations (changing these affects the types of edge interactions that one
obtains, if any.)
\end{enumerate}

\item Form the vector $B$ whose $i^{\mathrm{th}}$ entry is $%
B_{i}=a_{i}^{t}Qa_{i}$, where $a_{i}$ are the elements of the null space of $%
\mathcal{L}$. This is the left-hand side of Eq.~(\ref{eq:w}). Each entry of $%
B$ consists of linear equations whose variables $z_{i}$ represent the
presence or absence of a $Z$ operation in a quantum circuit that corresponds
to $G$.

\item Form a matrix $W$ whose rows are all possible bonds.

\item Produce a matrix $D$ whose $i^{\mathrm{th}}$ row $D_{i}$ is equal to $%
MW_{i}$. These are all possible values of the right hand side of Eq.~(\ref%
{eq:w}). (Note that due to symmetry there will be many repeats, so that the
total number of possible bonds to check is far fewer than all possible
bonds.)

\item Attempt to solve the system of linear equations $%
B_{k}=D_{1,k},B_{k}=D_{2,k},\dots ,B_{k}=D_{|L|,k}$ over $GF(2)$ (where $k$
runs from $1$ to the number of rows of $D$) for the variables $z_{i}$. A
solution for some $k$ gives information for a specific circuit
representation $H$. If no solution for the $z_{i}$ exists, then there is no
satisfying bond distribution $w$ for Eq.~(\ref{eq:w}). This is indeed the
case for $K_{3,3}$, $K_{5}$, $\bar{K}_{3,3}$ and $K_{4}$. If there is a
solution for some fixed $k$, continue.

\item Take this specific $H$ (i.e., this $H$ has no variables and
corresponds to a specific circuit given by the solution of the $z_{i}$
above) and now form $B_{i}=a_{i}^{t}(\mathrm{lwtr}H^{t}CH)a_{i}$ again. This
time however, $B$ contains no variables and is a numerical vector. Thus, one
now has the binary vector $B$ and the matrix $M$.

\item Solve the linear system $Mw=B$ and output the edge interaction $w$.
\end{enumerate}

Having applied this algorithm we proved using mathematical software that the
obstruction set for $\Theta $ includes $K_{3,3}$, $K_{5}$, $\bar{K}_{3,3}$
and $K_{4}$, i.e., Lemma \ref{lem:obs}.


\end{document}